\documentclass{amsart}
\pdfoutput=1

\usepackage[dvipsnames,pdftex]{xcolor}
\usepackage{amsmath,amssymb,amsfonts}
\usepackage{multirow}

\usepackage[noend]{algorithmic}
\usepackage{tabularx}
\usepackage{arydshln}
\usepackage[mathscr]{eucal}
\usepackage[ruled]{algorithm}
\usepackage[normalem]{ulem}
\usepackage{fancyhdr}
\usepackage{tikz}
\usepackage{url}
\usepackage[stable]{footmisc}

\usepackage[colorlinks=true,linktocpage]{hyperref}
\providecolor{DarkBlue}{rgb}{0,0,.545}
\providecolor{DarkGreen}{rgb}{0,.392,0}
\hypersetup{citecolor=DarkGreen}
\hypersetup{linkcolor=DarkBlue}
\hypersetup{urlcolor=DarkBlue}

\usepackage[nameinlink,capitalise]{cleveref}
\usepackage{multirow}
\usepackage[all]{xy}
\usepackage{footmisc}
\usepackage{graphicx}
\usepackage{tikz-cd}
\usepackage{xspace}

\usepackage{array}
\usepackage{booktabs}
\usepackage{makecell}
\usepackage{caption}
\usepackage{siunitx}
\usepackage{colonequals}

\newtheorem{theorem}{Theorem}[section]

\newtheorem{lemma}[theorem]{Lemma}

\newtheorem{problem}[theorem]{Problem}

\theoremstyle{definition}
\newtheorem{definition}[theorem]{Definition}
\newtheorem{example}[theorem]{Example}

\theoremstyle{remark}
\newtheorem{remark}[theorem]{Remark}

\theoremstyle{remark}

\newcommand{\SWITCH}[1]{\STATE \textbf{switch} (#1)}
\newcommand{\ENDSWITCH}{\STATE \textbf{end switch}}
\newcommand{\CASE}[1]{\STATE \textbf{case} #1\textbf{:} \begin{ALC@g}}
\newcommand{\ENDCASE}{\end{ALC@g}}
\newcommand{\CASELINE}[1]{\STATE \textbf{case} #1\textbf{:} }
\newcommand{\DEFAULT}{\STATE \textbf{default:} \begin{ALC@g}}
\newcommand{\ENDDEFAULT}{\end{ALC@g}}
\newcommand{\DEFAULTLINE}[1]{\STATE \textbf{default:} }

\newcommand{\field}{\ensuremath{\Bbbk}\xspace}
\newcommand{\fieldbar}{\ensuremath{\overline{\Bbbk}}\xspace}

\newcommand{\ZZ}{\ensuremath{\mathbb{Z}}}

\newcommand{\FF}{\ensuremath{\mathbb{F}}}
\newcommand{\QQ}{\ensuremath{\mathbb{Q}}}

\newcommand{\PP}{\ensuremath{\mathbb{P}}}
\newcommand{\AffineSpace}{\ensuremath{\mathbb{A}}}

\newcommand{\C}{\mathcal{C}}

\newcommand{\J}{\ensuremath{\mathcal{J}}\xspace}

\newcommand{\K}{\ensuremath{\mathcal{K}}\xspace}

\newcommand{\OO}{\ensuremath\mathcal{O}}

\DeclareMathOperator{\Jac}{Jac}

\newcommand{\Fp}{\mathbb{F}_p}
\newcommand{\Fpbar}{\overline{\mathbb{F}}_p} 
\newcommand{\Fpsq}{\mathbb{F}_{p^2}}

\newcommand\secp{\ensuremath{\lambda}\xspace}

\newcommand{\ourhash}{\textsf{KuHash}\xspace}

\newcommand{\Fqmult}{\texttt{M}\xspace}
\newcommand{\Fqsquare}{\texttt{S}\xspace}
\newcommand{\Fqadd}{\texttt{a}\xspace}

\newcommand{\PPAV}{p.p.a.v.\xspace}

 \usepackage[backend=bibtex, maxnames=99,natbib=true,sortcites]{biblatex} 

\bibliography{bib}

\begin{document}

\title{Efficient $(3,3)$-isogenies on fast Kummer surfaces}

\author{Maria Corte-Real Santos}
\address{University College London, UK}
\email{maria.santos.20@ucl.ac.uk}

\author{Craig Costello}
\address{Microsoft Research, USA}
\email{craigco@microsoft.com}

\author{Benjamin Smith}
\address{Inria and \'Ecole polytechnique, Institut Polytechnique de Paris, Palaiseau, France}
\email{smith@lix.polytechnique.fr}

\begin{abstract}
 We give an alternative derivation of $(N,N)$-isogenies between fast Kummer surfaces which complements existing works based on the theory of theta functions. We use this framework to produce explicit formul\ae{} for the case of $N=3$, and show that the resulting algorithms are more efficient than all prior $(3,3)$-isogeny algorithms. 
\end{abstract}

\maketitle

\section{
    Introduction
}
\label{sec:intro}

Isogenies of elliptic curves are well-understood,
at least from an algorithmic point of view,
in theory and in practice.
Given the Weierstrass equation of an elliptic curve \(E\),
and a generator \(P\) of a finite subgroup of \(E\),
Vélu's formul\ae{}~\cite{velu1971isogenies}
allow us to write down polynomials defining a normalized
quotient isogeny \(\Phi: E \to E/\langle P\rangle\)
(with variants for alternative curve
models~\cite{moody2016analogues}, or for rational subgroups with 
irrational generators~\cite{kohel1996endomorphism}).
Building on these formul\ae{},
there also exist highly efficient algorithms
for evaluating an isogeny at points of \(E\)
\emph{without} deriving a polynomial 
representation for the isogeny
itself (including~\cite{bernstein2020faster}, \cite{costello2017simple},
and~\cite{renes2018computing}, for example).
Interest in these formul\ae{} and algorithms has recently intensified
with the development of \emph{isogeny-based cryptography}
as a source of cryptosystems conjectured to be resistant against quantum attacks.

As a generalization of an elliptic curve, we consider principally
polarized abelian varieties,
and the first non-elliptic examples are Jacobians of genus-2 curves.
Genus-2 curves are hyperelliptic curves with affine plane models
\[
    C: y^2 = f(x)
    \quad
    \text{where \(f(x)\) is squarefree of degree 5 or 6}
    \,,
\]
in characteristic not dividing $30$, and the Jacobian \(\J = \Jac(C)\) of \(C\)
is a 2-dimensional principally polarized abelian variety
(\PPAV),
birational to \(C^{(2)}\),
parameterizing the degree-0 Picard group of the curve \(C\).

Mumford~\cite{mumford1984tata}, Cantor~\cite{cantor1987computing},
Grant~\cite{grant1990formal}, and Flynn~\cite{flynnthesis,Flynn90}
laid the ground for explicit geometric and number-theoretic computations
with genus-2 Jacobians. Cassels and Flynn's text~\cite{cassels1996prolegomena}
presents a unified view of genus-2 arithmetic.
Later,
Gaudry~\cite{gaudry} proposed Kummer surfaces of genus-2
Jacobians as a setting for efficient discrete-logarithm-based
cryptosystems, building on a variant~\cite{chudnovsky1986sequences} of Lentra's ECM factoring algorithm~\cite{lenstra1987factoring}.
The Kummer surface $\K$ of a Jacobian $\J$ is the image of the quotient
morphism $\pi: \J \rightarrow \K = \J/\langle\pm 1 \rangle$;
as such, it is the genus-2 analogue of the \(x\)-coordinate of elliptic
curves.
Geometrically, Kummer surfaces have convenient models as 
quartic surfaces in \(\PP^3\) with 16 point singularities.

Cosset put Chudnovsky and Chudnovsky's Kummer ECM into practice in~\cite{cosset2010factorization},
while high-speed, high-security Kummer-based implementations of Diffie--Hellman key
exchange~\cite{bos2016fast,bernstein2014kummer,renes2016kummer}
and signature schemes~\cite{renes2016kummer,qDSA}
can give significant practical improvements over elliptic curves
in many contexts.

However, while the basic arithmetic of genus-2 Jacobians and Kummer surfaces
has matured,
and while cryptographic applications have driven great improvements
in the efficiency of the resulting formul\ae{} and algorithms,
the corresponding explicit theory of isogenies lags behind.
First, note that just as elliptic isogenies factor naturally into
compositions of scalar multiplications 
and isogenies with prime cyclic kernel
(i.e., isomorphic to \(\ZZ/N\ZZ\) with \(N\) prime),
isogenies of abelian surfaces (including Jacobians of genus-2 curves)
decompose into compositions of scalar multiplications
and \emph{\((N,N)\)-isogenies} (with kernel isomorphic to
\((\ZZ/N\ZZ)^2\)).\footnote{%
    In some special cases,
    depending on the endomorphism ring of the Jacobian,
    we can also have isogenies with cyclic kernel~\cite{dudeanu22}.
    These isogenies are beyond the scope of this article.
}
The fundamental task, then,
is to compute and evaluate \((N,N)\)-isogenies
where \(N\) is prime.
We can do this on the level of the Jacobian
(using e.g. correspondences on genus-2 curves~\cite{Smith2006}),
or we can use the fact that isogenies commute with \(-1\)
to move down to the more tractable Kummer surfaces.
Indeed, as Cassels and Flynn note,
``we lose nothing by going down to the Kummers, because [the map] lifts
automatically to a map of abelian
varieties.''~\cite[\S9.3]{cassels1996prolegomena}.

The case \(N = 2\) is classical:
explicit methods and formul\ae{}
go back to Richelot~\cite{richelot1837transformatione},
and were re-developed in modern terms
by Bost and Mestre~\cite{bost1988moyenne}
and Cassels and Flynn~\cite[\S 3]{cassels1996prolegomena}.
Going further, we find some first efforts 
at explicit curve-based formul\ae{} for the case \(N = 3\) 
by Smith~\cite{smith2012computing}
(building on an ineffective general method due to Dolgachev and Lehavi~\cite{dolgachev2008isogenous}),
and more general results due to Couveignes and Ezome~\cite{couveignes2015computing}.
Moving to general Kummer surfaces,
Bruin, Flynn and Testa~\cite{bruin2014descent},
Nicholls~\cite{nicholls2018descent},
and Flynn~\cite{flynn2015descent}
gave more powerful formul\ae{} for \(N = 3\), \(4\), and \(5\),
respectively,
in a number-theoretic context;
Nicholls even gives a method for general \(N\).
Flynn and Ti revisited the formul\ae{} for \(N = 3\)
in a cryptographic context~\cite{g2SIDH},
and Decru and Kunzweiler~\cite{DecruKunzweiler23} further optimized
these formul\ae{}, drastically improving their efficiency. 
However, none of these formul\ae{} make use of the special symmetries of the most
efficient Kummer surfaces that have been used in cryptographic implementations.

Bisson, Cosset, Lubicz, and Robert have advanced an ambitious program~\cite{bisson2011endomorphism,cossetthesis,robertthesis,LubiczRobert,robertHDR}
based on the theory of theta functions~\cite{mumford1984tata}
to provide asymptotically efficient algorithms for arbitrary odd \(N\)
(and beyond genus~2 to arbitrarily high dimension). The \texttt{AVIsogenies} software package based on their results is
publicly available~\cite{AVIsogenies}.
These algorithms are certainly compatible with fast Kummer surfaces,
but they target isogeny evaluation for general abelian
varieties, rather than the construction of compact explicit formul\ae{}
in genus~2 that can be studied, analysed, and optimized in their own right.
Nevertheless,
these techniques were recently revisited by Dartois, Maino, Pope and Robert~\cite{22isostheta} in the
context of cryptography to efficiently compute chains of
$(2,2)$-isogenies between products of elliptic curves in the theta
model.

\subsection{Contributions.}

In this article,
we give a general method for deriving explicit
formul\ae{} for isogenies of fast Kummer surfaces,
optimizing the approach of Bruin, Flynn, and Testa
by exploiting the high symmetry of these ``fast'' surfaces,
which are the most relevant for applications over finite fields.
Our methods are elementary
in the sense that they avoid explicitly using the heavy machinery of theta functions
required in~\cite{cosset2013algorithm,lubicz2022fast,22isostheta}
(though of course theta functions implicitly play a fundamental role in our techniques).
We apply these methods to give explicit examples
for \(N = 3\) and \(5\). 
For example, for \(N = 3\) we obtain
a map \(\phi: \K\to\K'\)
defined by
\[
    \phi((X_1 : X_2 : X_3 : X_4)) = 
    \left(\phi_{1}(X_1,X_2,X_3,X_4)
    : 
    \cdots
    : \phi_{4}(X_1,X_2,X_3,X_4) \right)
    \,,
\]
where
\begin{align*}
    \phi_1(X_1,X_2,X_3,X_4) & = X_1 \left(c_{1}X_1^2+  c_2 X_2^2+  c_3 X_3^2+  c_4 X_4^2 \right) + c_5 X_2X_3X_4
    \,,
    \\
    \phi_2(X_1,X_2,X_3,X_4) & = X_2  \left(c_2 X_1^2+  c_{1} X_2^2+  c_4 X_3^2+  c_3 X_4^2 \right) + c_5 X_1X_3X_4
    \,,
    \\
    \phi_3(X_1,X_2,X_3,X_4) & = X_3  \left(c_3 X_1^2+  c_4 X_2^2+  c_{1} X_3^2+  c_2 X_4^2 \right) + c_5 X_1X_2X_4
    \,,
    \\
    \phi_4(X_1,X_2,X_3,X_4) & = X_4  \left(c_4 X_1^2+  c_3 X_2^2+  c_2 X_3^2+  c_{1} X_4^2 \right) + c_5 X_1X_2X_3
    \,,
\end{align*}
and $c_{i}$ are rational functions in the theta-null
constants $a,b,c,d$ defining \K 
and the coordinates of the generators of the kernel
(see~\cref{subsec:explicit-coeffs-33} for the explicit
expresssions).
This map can be evaluated with at most $88$ multiplications and $12$ squarings in
the field containing the theta constants and generator coordinates.

To illustrate the potential benefits of our formul\ae{}
in practical applications,
we give experimental results on cryptographic hash functions
based on chains of \((3,3)\)-isogenies,
as in~\cite{castryck2021multiradical} and~\cite{DecruKunzweiler23}.
In~\cref{sec:kernels} we present
\textsf{3DAC}: a three-dimensional
\emph{differential addition chain},
and use it
to construct $(N^k,N^k)$-isogeny kernels correctly, efficiently, and securely.
Combined with our $(3,3)$-isogeny formul\ae{},
this allows us to efficiently compute $(3^k,3^k)$-isogenies on fast Kummer surfaces.
The hash function we define in~\cref{sec:hashfunction} uses these
isogenies,
exploiting the efficient arithmetic of fast Kummer surfaces for the first time,
to gain speed-ups of between $8{\tt x}$ and $9{\tt x}$
over the Castryk--Decru hash function~\cite{castryck2021multiradical} 
and between $32{\tt x}$ and $34{\tt x}$ 
over the Decru--Kunzweiler hash function~\cite{DecruKunzweiler23}.

\subsection{Software.} The source code accompanying this paper is written in {\tt MAGMA}~\cite{MAGMA}, Python and SageMath~\cite{sagemath} and is publicly available under the MIT license. 
It is available at 
\begin{center}
    \url{https://github.com/mariascrs/KummerIsogenies}. 
\end{center}

\subsection*{Acknowledgements}
We thank Chris Nicholls for making the code accompanying his thesis available to us, which we used to obtain the formulae for $(2,2)$-isogenies presented in~\cref{sec:generalmethod}. 
We also thank Sam Frengley for many helpful conversations during the preparation of this paper. 
We are also grateful to the anonymous
reviewers, whose comments helped to significantly improve the article.
The first author was supported by UK EPSRC grant EP/S022503/1.
This work was supported by the HYPERFORM consortium,
funded by France through Bpifrance;
it also received funding from the France 2030 program,
managed by the French National Research Agency
under grant agreement No.~ANR-22-PETQ-0008 PQ-TLS,
and from the French National Research Agency
under grant agreement No.~ANR-19-CE48-0008 CIAO.

\section{
    Fast Kummer surfaces and their arithmetic
}
\label{sec:prelims}

Let \field be a perfect field---typically, a finite field or a number
field---of characteristic $p \neq 2, 3$, or $5$,
and fix an algebraic closure \fieldbar.
If \(\field = \FF_q\),
then we measure the time complexity of our algorithms
in terms of elementary operations in \(\FF_q\).
We let \Fqmult, \Fqsquare, and \Fqadd
denote the cost of a single multiplication, squaring,
and addition (or subtraction) in \(\FF_q\), respectively.

\subsection{Genus-2 curves and their Jacobians}
Every smooth genus-2 curve over \field is isomorphic to a curve of the
form $C: y^2 = f(x)$, where $f(x) \in \field[x]$ is a squarefree
polynomial of degree 6.

For our applications we may suppose that \(f(x)\) has all its roots in \field,
so (after, for example, mapping one root to \(0\), one to \(1\), and one to
\(\infty\), as in~\cite[\S 2.1]{ohashi2023rosenhain}) 
we can assume \(C\) has a \emph{Rosenhain model}
\[
    C \cong C_{\lambda, \mu, \nu}/\field: 
    y^2 = x(x-1)(x-\lambda)(x-\mu)(x-\nu)
    \quad
    \text{with }
    \lambda, \mu, \nu \in \field
    \,;
\]
the values \(\lambda\), \(\mu\), and \(\nu\)
are called \emph{Rosenhain invariants} of \(C_{\lambda,\mu,\nu}\).

Fix a prime \(N\) not divisible by \(\operatorname{char}(\field)\).
The $N$-torsion subgroup \(\J[N]\) of $\J$ is a $\ZZ/N\ZZ$-module of
rank 4,
that is,
$\J[N] \cong (\ZZ/N\ZZ)^4$;
and
the (canonical) principal polarisation on $\J$ induces a non-degenerate,
bilinear, and antisymmetric \emph{$N$-Weil pairing}
\[
    e_N: \J[N] \times \J[N] \rightarrow \mu_N
    \,.
\]

A subgroup $G \subseteq \J[N]$ is \emph{isotropic} if $e_N(P,Q) = 1$ for all $P,Q \in G$, 
and \emph{maximal isotropic} if it is not properly contained in any
isotropic subgroup of $\J[N]$.
Since \(N\) is prime,
if \(G\) is maximal isotropic then
$G \cong (\ZZ/N\ZZ)^2$; we say $G$ is an \emph{$(N,N)$-subgroup}.
If $G$ is a maximal isotropic subgroup of $\J[N]$, 
then the quotient isogeny of abelian varieties 
\[
    \Phi: \J \to A' := \J/G
\]
is an isogeny of \PPAV{}s:
that is,
there is a principal polarization on $A'$
that pulls back via $\Phi$ to \(N\) times the principal polarization on $\J$.
We say that $\Phi$ is an \emph{$(N,N)$-isogeny}.

Being a principally polarized abelian surface,
$A'$ is (as a \PPAV) the Jacobian of a genus 2 curve, say $\J'$,
or a product of elliptic curves $E'_1 \times E'_2$ 
(equipped with the product polarisation). 
The case \(A' = \J\) is the general case,
and the primary focus of this paper.

\subsection{Isogenies and Kummer surfaces}

The Kummer surface $\K$ of a Jacobian $\J$ is the image of the quotient map $\pi: \J \rightarrow \K = \J/\{\pm 1 \}$. 
Geometrically, it has a quartic model in $\PP^3$ with sixteen point
singularities, called \emph{nodes};
the nodes are the images in $\K$ of the 2-torsion points of $\J$,
since these are precisely the points fixed by $-1$. 

Any $(N,N)$-isogeny $\Phi: \J \rightarrow \J^\prime$ descends to a
morphism of Kummer surfaces $\phi: \K \rightarrow \K^\prime$, such that
the following diagram commutes:
\[\begin{tikzcd}
	{\J} && {\J^\prime} \\
	\\
	{\K} && {\K^\prime}
	\arrow["\Phi"', from=1-1, to=1-3]
	\arrow["\pi", from=1-1, to=3-1]
	\arrow["{\pi^\prime}"', from=1-3, to=3-3]
	\arrow["\phi", from=3-1, to=3-3]
\end{tikzcd}\]
Abusing terminology, we say a morphism $\phi$ of Kummer surfaces is an $(N,N)$-isogeny if it is
induced by an $(N,N)$-isogeny $\Phi$ between the corresponding Jacobians.

\subsection{Fast Kummer surfaces}\label{subsec:prelims:fastkummer}

Following Gaudry~\cite{gaudry}, fast Kummer surfaces
are defined by four \emph{fundamental theta
constants}, which can be computed from the
Rosenhain invariants of a genus-2 curve $C/\field$.
Given a hyperelliptic curve $C/\field$ with Rosenhain invariants 
$\lambda, \mu,\nu \in \field$, we define 
\emph{fundamental theta constants} $a,b,c,d \in \fieldbar$
and \emph{dual theta constants} as $A,B,C,D \in \fieldbar$ such that 
\begin{align*}
    A^2 &= a^2 + b^2 + c^2 + d^2, \ \ B^2 = a^2 + b^2 - c^2 - d^2, \\
    C^2 &= a^2 - b^2 + c^2 - d^2, \ \ D^2 = a^2 - b^2 - c^2 + d^2.
\end{align*}

The theta constants are
related to Rosenhain invariants through the relations
\[
    \lambda = \frac{a^2c^2}{b^2d^2}\,,
    \qquad
    \mu = \frac{c^2e^2}{d^2f^2}\,,
    \qquad
    \nu = \frac{a^2e^2}{b^2f^2}\,,
\]
where $e,f \in \fieldbar$ satisfy $e^2/f^2 = (AB+CD)/(AB-CD)$.

We define the \emph{fast} Kummer model \(\K\)
corresponding to \(\C\) as
\begin{equation}\label{eq:gaudrykummer}
    \begin{aligned}
      \K :  &\    X_1^4+X_2^4+X_3^4+X_4^4 - 2 E \cdot X_1X_2X_3X_4 -F\cdot (X_1^2X_4^2+X_2^2X_3^2) \\
              & - G\cdot (X_1^2X_3^2+X_2^2X_4^2)-H \cdot (X_1^2X_2^2+X_3^2X_4^2) = 0,
    \end{aligned}
\end{equation}
where \(X_1,X_2,X_3,X_4\) are coordinates on \(\PP^3\)
and the coefficients $E,F,G,H$ are rational functions in $a,b,c,d$, namely
\begin{equation}\label{eq:EFGH}
    \begin{aligned}
        E &:= 256abcdA^2B^2C^2D^2/(a^2d^2-b^2c^2)(a^2c^2-b^2d^2)(a^2b^2-c^2d^2),  \\
        F &:= (a^4-b^4-c^4+d^4)/(a^2d^2-b^2c^2),  \\
        G &:= (a^4-b^4+c^4-d^4)/(a^2c^2-b^2d^2),  \\
        H &:= (a^4+b^4-c^4-d^4)/(a^2b^2-c^2d^2).
    \end{aligned}
\end{equation}
This model of $\K$ is often referred to as the \emph{canonical}
parameterisation~\cite{qDSA}.
Note that $A^2$, $B^2$, $C^2$, and $D^2$
are linear combinations of $a^2$, $b^2$, $c^2$, and $d^2$,
so the equation of $\K$ is determined entirely by $a,b,c,d$;
in fact, \(\K\) is determined by the projective point $(a :  b :  c : d) \in \PP^3$.
The identity element on $\K$ is $\OO_\K = (a\colon b \colon c \colon d)$.

\subsection{Nodes of the Kummer surface}\label{subsec:nodes}
The \emph{nodes} of \(\K\) 
are the sixteen points
\begin{footnotesize}
\begin{align*}
    \label{eq:twotors}
    \OO_\K & = (a:b:c:d),
	&
    T_1 & = (a : b : -c : -d), 
    & 
    T_2 & = (a : -b : c : -d),
    & 
    T_3 & = (a : -b : -c : d),
    \\
    T_4 & = (b : a : d : c),
    & 
    T_5 & = (b : a : -d : -c),
    & 
    T_6 & = (b : -a : d : -c),
    & 
    T_7 & = (b : -a : -d : c),
    \\
    T_8 & = (c : d : a : b), 
    & 
    T_9 & = (c : d : -a : -b),
    & 
    T_{10} & = (c : -d : a : -b),
    &
    T_{11} & = (c : -d : -a : b),
    \\
    T_{12} & = (d : c : b : a),
    &
    T_{13} & = (d : c : -b : -a),
    &
    T_{14} & = (d : -c : b : -a),
    & 
    T_{15} & = (d : -c : -b : a).
\end{align*}
\end{footnotesize}
Each of the \(T_i\) is the image in \(\K\)
of a two-torsion point \(\widetilde{T}_i\) in \(\J[2]\).
Since \(\tilde{T}_i = -\widetilde{T}_i\),
the translation-by-\(\widetilde{T}_i\) map on \(\J\)
induces a morphism \(\sigma_i: \K \to \K\).
In fact, \(\sigma_i\) lifts to a linear map on \(\AffineSpace^4\): that is, it acts like a matrix on the coordinates
\((X_1,X_2,X_3,X_4)\) on \(\PP^3\).
Further, \(\sigma_i\) and \(\sigma_j\)
commute resp.~anticommute
if \(e_2(\widetilde{T}_i,\widetilde{T}_j) = 1\)
resp.~\(-1\).

In particular, if we define
\begin{align*}
    U_1 & := \operatorname{diag}(1,1,-1,-1),
    &
    U_2 & := \operatorname{diag}(1,-1,1,-1),
\end{align*}
and
\begin{align*}
    V_1 & := \begin{pmatrix}
        0 & 1 & 0 & 0
        \\
        1 & 0 & 0 & 0
        \\
        0 & 0 & 0 & 1
        \\
        0 & 0 & 1 & 0
    \end{pmatrix}
    \,,
    &
    V_2 & := \begin{pmatrix}
        0 & 0 & 0 & 1
        \\
        0 & 0 & 1 & 0
        \\
        0 & 1 & 0 & 0
        \\
        1 & 0 & 0 & 0
    \end{pmatrix}
    \,.
\end{align*}
Then
\begin{equation}
    \label{eq:U-V-relations-1}
    U_1^2 = U_2^2 = I_4, \text{ and } U_1U_2 = U_2U_1
    \,,
    \qquad
    V_1^2 = V_2^2 = I_4, \text{ and } V_1V_2 = V_2V_1
    \,,
\end{equation}
and
\begin{equation}
    \label{eq:U-V-relations-2}
    U_1V_2 = -V_2U_1
    \,,
    \qquad
    U_2V_1 = -V_1U_2
    \,,
    \qquad
    U_1V_1 = V_1U_1
    \,,
    \qquad
    U_2V_2 = V_2U_2
    \,.
\end{equation}
Taking the labelling of the nodes above,
with \(T_0 = (a:b:c:d)\) as the image of \(\widetilde{T}_0 = \OO_\J\),
the corresponding translations are
such that
\[
    T_i = \sigma_i((a:b:c:d)) 
    \qquad
    \text{for }
    0 \le i \le 15
    \,;
\]
that is,
\begin{align*}
    \sigma_0    & = I_4
    & 
    \sigma_1    & = U_1
    & 
    \sigma_2    & = U_2
    & 
    \sigma_3    & = U_1U_2
    \\
    \sigma_4    & = V_1 
    &
    \sigma_5    & = V_1U_1
    &
    \sigma_6    & = V_1U_2
    &
    \sigma_7    & = V_1U_1U_2
    \\
    \sigma_8    & = V_1V_2
    &
    \sigma_9    & = V_1V_2U_1
    &
    \sigma_{10} & = V_1V_2U_2
    &
    \sigma_{11} & = V_1V_2U_1U_2
    \\
    \sigma_{12} & = V_2
    &
    \sigma_{13} & = V_2U_1
    &
    \sigma_{14} & = V_2U_2
    &
    \sigma_{15} & = V_2U_1U_2
\end{align*}
Now Eqs.~\eqref{eq:U-V-relations-1}
and~\eqref{eq:U-V-relations-2}
show that
\((\widetilde{T}_1, \widetilde{T}_2, \widetilde{T}_{12}, \widetilde{T}_4)\)
is a symplectic basis (with repesct to the \(2\)-Weil pairing) of~\(\J[2]\), 
where we define a symplectic basis as follows. 

\begin{definition}\label{def:symplectic}
    Let $\J$ be the Jacobian of a genus $2$ curve $\C$. 
    We say that a basis $\{Q_1, Q_2, Q_3, Q_4\}$ for $\J[D]$ 
    is \emph{symplectic} with respect to the $D$-Weil pairing
    if 
	\[
		e_D(Q_1, Q_3) = e_D(Q_2, Q_4) = \zeta
	\]
	where $\zeta$ is a primitive $D$-th root of unity, and $e_D(Q_i, Q_j) = 1$ otherwise.
\end{definition}

\subsection{Operations on the Kummer surface}\label{sub:operations}
Let $\pi: \J \to \K$ be the quotient by $-1$.  
The multiplication-by-$m$ maps $[m]$ on $\J$ induce 
\emph{pseudo-multiplications} $\pi(P) \mapsto [m]_*(\pi(P)) = \pi([m]P)$.  
We can express the pseudo-doubling map $[2]_*$ on $\K$ as a composition 
of four basic building blocks,
each a morphism from \(\PP^3\) to \(\PP^3\):
\begin{enumerate}
    \item
        the \emph{Hadamard involution} $\mathcal{H}: \PP^3 \to \PP^3$,
        which is induced by the linear map on \(\AffineSpace^4\) defined by 
        the matrix
        \[
            \left(
            \begin{array}{rrrr}
                1  & 1  &  1 &  1
                \\
                1  & 1  & -1 & -1
                \\
                1  & -1 &  1 & -1
                \\
                1  & -1 & -1 &  1
            \end{array}
            \right)
            \,;
        \]
    \item
        the \emph{squaring} map
        \[
            \mathcal{S}: (X_1:X_2:X_3:X_4)
            \longmapsto
            (X_1^2:X_2^2:X_3^2:X_4^2)
            \,;
        \]
    \item
        the \emph{scaling} maps 
        \[
            \mathcal{C}_{(\alpha : \beta : \gamma : \delta)}
            : 
            (X_1 : X_2 : X_3 : X_4) 
            \longmapsto 
            (\alpha X_1 : \beta X_2 : \gamma X_3 : \delta X_4)
        \]
        for each \((\alpha:\beta:\gamma:\delta) \in \PP^3(\field)\);
        and
    \item
        the \emph{inversion} map 
        \begin{align*}
            \mathcal{I}: 
            (X_1 : X_2 : X_3 : X_4) 
            \longmapsto 
            & \ 
            (X_2X_3X_4: X_1X_3X_4: X_1X_2X_4: X_1X_2X_3)
            \\
            =
            & \ 
            (1/X_1 : 1/X_2 : 1/X_3 : 1/X_4)
            \,,
        \end{align*}
        well-defined when all $X_i \neq 0$.
\end{enumerate}
We can readily see that 
$\mathcal{H}$ costs 8 \field-additions,
$\mathcal{S}$ costs 4 \field-squarings,
$\mathcal{C}$ costs 4 \field-multiplications,
and
$\mathcal{I}$ costs 6 \field-multiplications.
Now, 
if $\K$ is a Kummer surface with fundamental theta constants $(a :  b :  c :  d)$,
then pseudo-doubling is given by 
\[
    [2]_*
    = 
    \mathcal{C}_{\mathcal{I}(\mathcal{O}_{\K})} \circ \mathcal{H} \circ \mathcal{S} \circ \mathcal{C}_{\mathcal{I}((A:  B :  C :  D))} \circ \mathcal{H} \circ \mathcal{S}
    \,.
\]

While \K inherits scalar multiplication from \J,
it loses the group law:
$\pi(P) = \pm P$ and $\pi(Q) = \pm Q$ in $\K$
do not uniquely determine
$\pi(P+Q) = \pm(P+Q)$ (unless at least one of \(P\) and \(Q\) is in \(\J[2]\)).
However, 
the operation $\{ \pi(P), \pi(Q) \} \mapsto \{ \pi(P+Q), \pi(P-Q) \}$ is well-defined,
so we have
a \emph{pseudo-addition} operation
$(\pi(P), \pi(Q), \pi(P-Q))\mapsto\pi(P+Q)$.

By abuse of notation, we let $R=(r_1 \colon r_2 \colon r_3 \colon r_4)$ and $S=(s_1 \colon s_2 \colon s_3 \colon s_4)$ be points on $\K$,
and let $T^{+}=(t_1^+ \colon t_2^+ \colon t_3^+ \colon t_4^+)$ and $T^{-}=(t_1^- \colon t_2^- \colon t_3^- \colon t_4^-)$ denote the sum $R+S$ and difference $R-S$, respectively.  
There exist biquadratic forms $B_{ij}$~\cite[Theorem 3.9.1]{cassels1996prolegomena} for $\K$ such that for $1 \leq i,j \leq 4$ we have
$$
    t^{+}_i t^{-}_j + t^{-}_i t^{+}_j 
    = 
    \lambda B_{ij}\big(r_1,r_2,r_3,r_4; s_1,s_2,s_3,s_4\big) 
    = 
    \lambda B_{ij}(R; S),
$$
where $\lambda \in \fieldbar$ is a common projective factor depending only on the affine representations chosen for $R$, $S$, $T^{+}$, $T^-$. 
The biquadratic forms \(B_{i,j}\) for fast Kummer surfaces
are given explicitly by Renes and Smith~\cite[\S 5.2]{qDSA}. 

These biquadratic forms are the basis of explicit pseudo-addition and
doubling laws on $\K$. For example, if the difference $T^{-}$ is known,
then the \(B_{ij}\) can be used to compute the coordinates of $T^{+}$.
As we will see in~\cref{sec:generalmethod}, the \(B_{ij}\)
will also be crucial in determining equations for our $(N, N)$-isogenies. 

\section{
    \texorpdfstring{$(N,N)$}{(N,N)}-isogenies on fast Kummer surfaces
}
\label{sec:generalmethod}

Throughout this section,
$\Phi: \J = \Jac(C) \to \J^\prime$ 
is
an $(N,N)$-isogeny with kernel $G \subset \J[N]$ 
(a maximal $N$-Weil isotropic subgroup of $\J[N]$),
where $N$ is 
a prime number not equal to the characteristic of the base field $\field$.
Our goal is to compute an explicit and efficiently-computable collection
of polynomials defining the induced map
$\phi: \K \to \K^\prime$
when $\K$ and $\K^\prime$ admit \emph{fast} models.

\subsection{A warm-up with \texorpdfstring{$N = 2$}{N = 2}}\label{subsec:warmup-22}

We first dispose of the case \(N = 2\).
Let $\widetilde{T}_0, \dots, \widetilde{T}_{15}$ be the 16 points in $\J[2]$, and let $T_0, \dots, T_{15}$ be their images in $\K$. Recall that $T_i = \sigma_i(T_0)$, where $\sigma_i: \K \rightarrow \K$ is a morphism defining the translation-by-$T_i$ map. 
There are precisely fifteen (images of) $(2,2)$-subgroups in $\K$,
and they are the images of
\[
    \widetilde{G}_{ij} 
    := 
    \big\{
        \widetilde{T}_0, 
        \ \widetilde{T}_i,
        \ \widetilde{T}_j,
        \ \widetilde{T}_i + \widetilde{T}_j
    \big\}
    \subset
    \J[2]
\]
in $\K$, where $1 \leq i\neq j \leq 15$ such that the linear maps
corresponding to $\sigma_i$ and $\sigma_j$ commute.
This gives 15 corresponding $(2,2)$-isogenies given by $\phi: \K\rightarrow \K' = \K/G_{i,j}$.
For each unique $(2,2)$-subgroup, we can associate a morphism $\alpha: \K \rightarrow \K$ induced by a linear map on $\mathbb{A}^4$ such that the corresponding $(2,2)$-isogeny $\K \rightarrow \widetilde{\K}$ is given by 
\[
    \psi := \mathcal{H} \circ \mathcal{S} \circ \alpha.
\]
We give matrices specifying the linear map $\alpha$ for each of the $(2,2)$-subgroups in~\cref{appendix:2-2}. For now, it suffices to note that the nonzero entries of these $4\times 4$ matrices are all fourth roots of unity.

Note however, that $\widetilde{\K}$ will not be in the correct form, as given by Equation~\eqref{eq:gaudrykummer}. We must therefore apply a final scaling 
$\mathcal{C}_{U}$ where $U = (\mathcal{S}^{-1}\circ \mathcal{I}\circ \psi)(\OO_{\K})$.
From this we obtain a $(2,2)$-isogeny $\phi := \mathcal{C}_{U}\circ \psi : \K \rightarrow \K' = \K/G_{i,j}$, where $\K$ and $\K'$ are fast Kummer surfaces. 

\begin{example}
    Consider the $(2,2)$-subgroup 
    \[
    G_{1,2} = \{ (a \colon b \colon c \colon d), (a \colon b \colon -c \colon -d), (a \colon -b \colon c \colon -d), (a \colon -b \colon -c \colon d)\}.
    \]
    Here $\alpha$ is the identity map and the $(2,2)$-isogeny is given by 
    \begin{equation*}
        \begin{aligned}
            (X_1 \colon X_2 \colon &X_3 \colon X_4) \mapsto \\
            &\Bigg(\frac{X_1^2+X_2^2+X_3^2+X_4^2}{A} \colon \frac{X_1^2+X_2^2-X_3^2-X_4^2}{B} \colon \\  
            & \hspace{25ex} \colon\frac{X_1^2-X_2^2+X_3^2-X_4^2}{C} \colon \frac{X_1^2-X_2^2-X_3^2+X_4^2}{D}\Bigg)
        \end{aligned}
    \end{equation*}
    We call this the \emph{distinguised kernel}: the kernel of the first half of doubling. Indeed, $U = (A: B : C : D)$ 
    and we recover the first three steps $\mathcal{C}_{\mathcal{\mathcal{I}}(A: B : C : D)} \circ \mathcal{H} \circ \mathcal{S}$ of the doubling map on fast Kummer surfaces. 
\end{example}

\begin{remark}
    Though the formul\ae{} for $(2,2)$-isogenies on fast Kummer surfaces are extremely 
    compact, we remark that the final scaling requires the computation of square roots 
    in \fieldbar. If $R$, $S$ are the $2$-torsion points generating the isogeny, 
    let $P$, $Q\in \K$ be such that $R = [2]_*P$, $S = [2]_*Q$. 
    The coordinates of $P$, $Q$ contain the square roots needed for the final scaling 
    (up to a projective factor). 
    We therefore suspect that these square roots may be inferred directly from coordinates 
    of the $4$-torsion, though we have not been able to derive such formul\ae{}.
\end{remark}

\subsection{The general case: odd \texorpdfstring{\(N\)}{N}}

From this point forward, we suppose \(N\) is an odd prime.
Since \(N\) is odd,
the \((N,N)\)-isogeny \(\Phi\) restricts to an isomorphism of
\(2\)-torsion subgroups \(\J[2]\to\J'[2]\). Furthermore, since \(\Phi\) is an isogeny of \PPAV
it is compatible with the \(N\)-Weil pairing by definition,
and so it maps the symplectic structure on \(\J[2]\) 
associated with the fast Kummer \(\K\)
onto a symplectic \(2\)-torsion structure on \(\J'[2]\),
which is associated with a fast Kummer \(\K'\).
The isogeny \(\Phi\)
therefore descends to a morphism
\(\phi: \K \to \K'\) of fast Kummers. Our goal is to construct explicit equations for \(\phi\).

To do so, we follow the strategy
taken by Cassels and Flynn~\cite[\S 9]{cassels1996prolegomena}, Bruin,
Flynn and Testa~\cite{bruin2014descent}, Nicholls~\cite[\S
5]{nicholls2018descent}, and Flynn~\cite{flynn2015descent},
adapting it to the case of \emph{fast} Kummer surfaces.
We will observe that the special forms of the affine translation maps
\(\sigma_i\) are very helpful in this setting,
and lead to nice results.

In practice,
we are given a fast Kummer \(\K\)
and the image \(\pi(G)\) of an \((N,N)\)-subgroup \(G\) of \(\J\) in~\(\K\).
There exists a fast Kummer
\(\K' \cong (\J'/G)/\langle{\pm1}\rangle\),
and our goal is to find \(\K'\) and the map \(\phi:\K\to\K'\) induced by the
quotient \((N,N)\)-isogeny \(\Phi\) with kernel \(G\).
Crucially, \(\phi\) ``commutes'' with the action by
\(2\)-torsion points, in the sense of the following definition.

\begin{definition}
    \label{def:fast-isogeny}
    An \emph{isogeny of fast Kummer surfaces}
    is a morphism \(\phi: \K \to \K'\)
    induced by an isogeny \(\Phi: \J \to \J'\)
    such that when lifted to a map on the ambient space,
    we have
    \[
        \phi\circ U_i^{\K} = U_i^{\K'}\circ \phi
        \quad
        \text{and}
        \quad
        \phi\circ V_i^{\K} = V_i^{\K'}\circ \phi
        \quad
        \text{for }
        i = 1, 2
        \,,
    \]
    where $U_i$ and $V_i$ are as defined in~\cref{subsec:nodes}.
\end{definition}

We want to compute \(\phi: \K \to \K'\),
but \(\K'\) is unknown.
However, as \(\K\) and \(\K'\) are both embedded in \(\PP^3\), \(\phi\) must be defined by forms of degree \(N\),
and it must commute with the actions of the \(U_i\) and \(V_i\).
This imposes heavy constraints on the shape of the forms defining \(\phi\),
and we can hope to interpolate them
using linear algebra given the action of \(G\),
and therefore to interpolate the image Kummer \(\K'\)
by pushing the theta constants $(a \colon b \colon c \colon d)$ through the isogeny.

Let $\K[N]$ be the image of $\J[N]$ on $\K$ and fix $R,S \in \K[N]$. From this point forward, we write $\langle R,S\rangle \subset \K[N]$ for the image of the subgroup $G$ of $\J[N]$ generated by the preimages of $R,S$. By abuse of notation, we say that $R$, $S$ are $N$-torsion points on $\K$.

The first step is to compute two sets of homogeneous forms of degree $N$ in the coordinates of $\K$ that are invariant under translation by $R$ and by $S$. 
The following lemma, due to Nicholls~\cite[\S 5.8.4]{nicholls2018descent}, describes how we can use the biquadratic forms associated to the Kummer surface introduced in~\cref{subsec:prelims:fastkummer} to construct these homogeneous forms. 

\begin{lemma}\label{lemma:invaraintforms}
    Fix Kummer surface $\K$ with coordinates $X_1, X_2, X_3, X_4$ 
    and associated biquadratic forms $B_{i,j}$ for $1 \leq i,j \leq 4$.
    Let \(N\) be an odd prime number,
    and fix a point $R \in \K[N]$ of order $N$.

    We denote by $I \in \{1,2,3,4\}^N$ a list of indices 
    $I = (i_1, \dots, i_N)$. 
    Letting $\tau$ be a permutation of $\{1, \dots, N\}$, we write $$\tau(I) =  \tau((i_1, \dots, i_N)) := (i_{\tau(1)}, \dots, i_{\tau(N)}).$$
    Then, for each $I \in \{1,2,3,4\}^N$ we define
    \[
        F_I 
        := 
        \sum_{\tau \in C_N} X_{i_{\tau(1)}}\cdot \prod_{k=1}^{(N-1)/2}B_{i_{\tau(2k)}, i_{\tau(2k+1)}}(X_1, X_2, X_3, X_4; \ kR)
        \,,
    \]
    where $C_N$ is the cyclic group of order $N$.
    Then, the set $\mathcal{F}_R := \{ F_I \}$ contains homogeneous forms of degree $N$ invariant under translation by $R$.
\end{lemma}

Applying the lemma above to $N$-torsion points $R$ and $S$, we obtain the two sets $\mathcal{F}_R$ and $\mathcal{F}_S$. 
The homogeneous forms of degree $N$ in each set will generate a space of dimension $m_N \geq 4$. 
Our experiments suggest $m_N = 2N+2$ for $N \leq 19$, and possibly beyond, 
though we have not proven this.

The next step is to compute a basis for these two spaces, say 
$F^R_1, \dots, F^R_{m_N}$ is a basis for the space generated by the homogeneous forms in $\mathcal{F}_R$, and $F^S_1, \dots, F^S_{m_N}$ a basis for the space generated by $\mathcal{F}_S$.

The intersection of these spaces contains homogeneous forms of degree $N$ that are invariant under translation by any point in the kernel $G$ of our $(N,N)$-isogeny. The intersection will be of dimension 4, and a basis for this intersection gives an $(N,N)$-isogeny $\psi: \K \rightarrow \widetilde{\K}$. Explicitly, the third step is to compute a basis of this 
intersection, say $f_1$, $f_2$, $f_3$, $f_4$. Then, our $(N,N)$-isogeny is given by $\psi = (f_1 : f_2 : f_3 : f_4)$. 

We note that $\widetilde{\K}$ may not be in the correct form given by Equation~\eqref{eq:gaudrykummer}. When computing chains of isogenies, however, it is important to ensure that our $(N,N)$-isogenies have domain and image in the same form.
Therefore, the last step is to apply a linear transformation $\textbf{M}: \widetilde{\K} \rightarrow \K^\prime$, where $\K'$ is a fast Kummer surface. Post-composing the map $\psi$ with this linear transformation gives a $\fieldbar$-rational $(N,N)$-isogeny $\phi: \K \rightarrow \K^\prime$ between fast Kummer surfaces generated by kernel $G = \langle R,S \rangle$.

\begin{remark}
    The compactness and efficiency of our isogeny formulae is determined by the choice of basis we make for the spaces generated by the forms in $\mathcal{F}_R$ and $\mathcal{F}_S$. An open question that arises from this work, therefore, is finding a solution to the following problem: let $f_1, \dots, f_n \in \QQ(a_1, \dots, a_k)[x_1, \dots, x_m]$ be a basis of polynomials defined over a function field. Find a ``nice'' basis $g_1, \dots, g_n$ where $g_1,\dots,g_n$ are $\QQ(a_1, \dots, a_k)$-linear combinations of the $f_i$. 
\end{remark}

\section{
    Explicit \texorpdfstring{$(3,3)$}{(3,3)}-isogenies on fast Kummers
}
\label{sec:3-3}

We now specialise the discussion in~\cref{sec:generalmethod} to $N = 3$ to construct $(3,3)$-isogenies between fast Kummer surfaces. 

Let $\J = \Jac(C)$ be the Jacobian of a genus 2 curve $C$ defined over $\fieldbar$.  
Suppose we have a $(3,3)$-subgroup of $\J[3]$, which induces an 
isogeny $\phi$ on the corresponding fast Kummer surface $\K = \J/\{\pm 1\}$ with kernel $G = \langle R,S \rangle$ for some $R,S \in \K[3]$ (i.e., $G$ is the image of the $(3,3)$-subgroup in $\K$).

Exploiting the fact that $\phi$ is an isogeny of fast Kummer surfaces, we obtain the following lemma, demonstrating that it is determined by five $\fieldbar$-rational functions in the coordinates of $\OO_{\K} = (a\colon b \colon c \colon d)$, $R$ and $S$.  
\begin{lemma}\label{lemma:33iso-form}
    Let $R$ and $S$ be distinct $3$-torsion points on $\K$ generating a $(3,3)$-subgroup $G \subset \K[3]$, and set $\OO_{\K} = (a \colon b \colon c \colon d)$.
    The $(3,3)$-isogeny of fast Kummer surfaces $\phi \colon \K \to \K^\prime$ generated by kernel $G$
    is in the form
    \begin{equation*}
        (X_1 : X_2 : X_3 : X_4)
        \longmapsto 
        \left(
            \phi_{1}(X_1,X_2,X_3,X_4) 
            : \cdots
            : \phi_{4}(X_1,X_2,X_3,X_4)
        \right)
        \,,
    \end{equation*}
    where 
    \begin{align*}
        \phi_1(X_1,X_2,X_3,X_4) & = X_1 \left(c_{1}X_1^2+  c_2 X_2^2+  c_3 X_3^2+  c_4 X_4^2 \right) + c_5 X_2X_3X_4
        \,,
        \\
        \phi_2(X_1,X_2,X_3,X_4) & = X_2  \left(c_2 X_1^2+  c_{1} X_2^2+  c_4 X_3^2+  c_3 X_4^2 \right) + c_5 X_1X_3X_4
        \,,
        \\
        \phi_3(X_1,X_2,X_3,X_4) & = X_3  \left(c_3 X_1^2+  c_4 X_2^2+  c_{1} X_3^2+  c_2 X_4^2 \right) + c_5 X_1X_2X_4
        \,,
        \\
        \phi_4(X_1,X_2,X_3,X_4) & = X_4  \left(c_4 X_1^2+  c_3 X_2^2+  c_2 X_3^2+  c_{1} X_4^2 \right) + c_5 X_1X_2X_3
        \,,
    \end{align*}
    with $c_{i}\in \fieldbar[a,b,c,d,r_1,r_2,r_3,r_4,s_1,s_2,s_3,s_4]$. 
\end{lemma}
\begin{proof}
    By \Cref{lemma:invaraintforms}, the isogeny $\phi$ is given by cubic forms.
    That is,
    $\phi$ is defined by polynomials
    \begin{align*}
        \phi_i = \  & c_{i,1}X_1^3 + c_{i,2}X_1X_2^2 + c_{i,3}X_1X_3^2 + c_{i,4}X_1X_4^2 + c_{i,5}X_2X_3X_4
        \\
        \quad + \ & c_{i,6}X_2X_1^2 + c_{i,7}X_2^3 + c_{i,8}X_2X_3^2 + c_{i,9}X_2X_4^2 + c_{i,10}X_1X_3X_4
        \\
        \quad + \ & c_{i,11}X_3X_1^2 + c_{i,12}X_3X_2^2 + c_{i,13}X_3^3 + c_{i,14}X_3X_4^2 + c_{i,15}X_1X_2X_4
        \\
        \quad + \ & c_{i,16}X_4X_1^2 + c_{i,17}X_4X_2^2 + c_{i,18}X_4X_3^2 + c_{i,19}X_4^3 + c_{i,20}X_1X_2X_3
        \,,
    \end{align*} 
    where $c_{i,j}$ are $\fieldbar$-rational functions in the
    coordinates of $\OO_{\K}$, $R$, and $S$ for $1 \leq i \leq 4$ and $1 \leq j \leq 20$.
    We are looking for an isogeny of fast Kummer surfaces in the sense
    of Definition~\ref{def:fast-isogeny},
    and compatibility with the translations by \(2\)-torsion thus forces
    \begin{align}
        \label{eq:lemma41}
        \sigma'_i((\phi_1 \colon \phi_2 \colon \phi_3 \colon \phi_4)) = \phi(\sigma_i(X_1 \colon X_2 \colon X_3 \colon X_4))
    \end{align}
    for all $1 \leq i \leq 15$,
    where $\sigma_i$ is the action of the $2$-torsion point $T_i \in \K$, and similarly $\sigma'_i$ is the action of $T'_i \in \K'$ (as defined in~\cref{subsec:nodes}).
    Equation~\eqref{eq:lemma41} gives rise to relations between the coefficients of the cubic monomials, from which we deduce that $\phi$ is of the form as in the statement of the lemma. See {\tt section4/lemma-4\_1.m} in the code accompanying this paper. 
    Clearing denominators (as our Kummer surfaces lie in $\PP^3$), we obtain the $c_i \in \fieldbar[a,b,c,d,r_1,r_2,r_3,r_4,s_1,s_2,s_3,s_4]$. 
\end{proof}

By~\cref{lemma:33iso-form}, to determine explicit formulae for the $(3,3)$-isogeny $\phi$ generated by kernel $G = \langle R,S \rangle \subset \K$, it suffices to determine the coefficients $c_1, \dots c_5$.
We follow the method given in~\cref{sec:generalmethod} and compute the $G$-invariant cubic forms.
Define
\begin{align*}
    B^R_{ij}(X_1, X_2, X_3, X_4) &\colonequals B_{i,j}(X_1, X_2, X_3, X_4; R), \\
    B^S_{ij}(X_1, X_2, X_3, X_4) &\colonequals B_{i,j}(X_1, X_2, X_3, X_4; S).
\end{align*}
By~\cref{lemma:invaraintforms}, the cubic forms invariant under translation by $R$ and $S$ are given by
\begin{align*}
	F^R_{ijk} &\colonequals X_iB^R_{jk} + X_jB^R_{ki} + X_kB^R_{ij}, \\ 
    F^S_{ijk} &\colonequals X_iB^S_{jk} + X_jB^S_{ki} + X_kB^S_{ij},
\end{align*}
respectively, where $1 \leq i,j,k \leq 4$. 
Let $\mathcal{F}_R = \{F^R_{ijk}\}_{1 \leq i,j,k \leq 4}$ and similarly define $\mathcal{F}_S$. 
The cubic forms in $\mathcal{F}_R$ and $\mathcal{F}_S$ each generate a space of dimension 8, for which we choose a basis $$\{F^R_{111}, \allowbreak \ F^R_{234}, \allowbreak \ F^R_{222}, \allowbreak \ F^R_{134}, \allowbreak \ F^R_{333}, \allowbreak \ F^R_{124}, \allowbreak \ F^R_{444}, \allowbreak \ F^R_{123}\},$$ and similarly for $\mathcal{F}_S$. 
These spaces will intersect in a space of dimension 4, which will give a description of the $(3,3)$-isogeny. 
We compute a basis
\begin{align*}
    f_1 & := z_1F^R_{111} + z_2F^R_{234}\,, 
    & 
    f_2 & := z_3F^R_{222} + z_4F^R_{134}\,,
    \\
    f_3 & := z_5F^R_{333} + z_6F^R_{124}\,,
    &
    f_4 & := z_7F^R_{444} + z_8F^R_{123}
    \intertext{for the intersection, with $z_1, \dots, z_8 \in \fieldbar$
    such that there exist $w_1, \dots, w_8 \in \fieldbar$ with}
    f_1 &= w_1F^S_{111} + w_2F^S_{234}\,,
    &
    f_2 &= w_3F^S_{222} + w_4F^S_{134}\,,
    \\
    f_3 &= w_5F^S_{333} + w_6F^S_{124}\,,
    &
    f_4 &= w_7F^S_{444} + w_8F^S_{123}\,.
\end{align*}
From this, we obtain a $(3,3)$-isogeny $\psi = (f_1\colon f_2 \colon f_3\colon f_4): \K \rightarrow \widetilde{\K}$. 

To move $\widetilde{\K}$ to the correct form, we first define 
\begin{align*}
    D_1 := (ab-cd)(ab+cd), \  D_2 := (ac-bd)(ac+bd), \ D_3 := (ad-bc)(ad+bc).
\end{align*}
For a point $P = (x_1\colon x_2 \colon x_3 \colon x_4)$, we define
$$\gamma(P) := (D_{23}(x_1x_2ab-x_3x_4cd)+D_{13}(x_1x_3ac-x_2x_4bd)+D_{12}(x_1x_4ad-x_2x_3bc)),$$
and $h_i(P)$ as the coordinates of $(\mathcal{H}\circ\mathcal{S})(P)$, for $i=1,2,3,4$.

Applying a linear transformation \textbf{M} to $\widetilde{\K}$, where \textbf{M} is defined as
\begin{equation*}
    \mathbf{M} := \begin{pmatrix}
        1/\alpha_1 & 0 & 0 & 0 \\
        0 & 2/\alpha_2 & 0 & 0 \\
        0 & 0 & 2/\alpha_1 & 0 \\
        0 & 0 & 0 & 2/(3\alpha_2) \\
    \end{pmatrix}
\end{equation*}
and where
\begin{align*}
    \alpha_1 &:= D_3(\gamma(R)(s_4s_3ab-s_1s_2cd) - \gamma(S)(r_4r_3ab-r_1r_2cd)),\\
    \alpha_2 &:= D_1(\gamma(R)(s_2s_3ad-s_1s_4bc)-\gamma(S)(r_2r_3ad-r_1r_4bc)),
\end{align*}
we get a simple and efficiently computable expression for our $(3,3)$-isogeny $\phi := \textbf{M}(f_1, f_2, f_3, f_4)^T$, whose image is in the desired form. 
Formul\ae{} for the intersection and $\textbf{M}$ can be found and verified in the file {\tt section4/linear-transform.m} in the accompanying code.

\subsection{Explicit formulae for \texorpdfstring{$(3,3)$}{(3,3)}-isogenies.}\label{subsec:explicit-coeffs-33} 

We now give explicit formul\ae{} for the isogeny $\phi : \K \rightarrow
\K^\prime$ with kernel $G = \langle R,S \rangle$.
By~\cref{lemma:33iso-form}, it suffices to give explicit formulae for the coefficients $c_i \in \fieldbar[a,b,c,d,r_1,r_2,r_3,r_4,s_1,s_2,s_3,s_4]$ for $i \in \{1,2,3,4,5\}$. We set
\begin{align*}
    \beta_1 &:= D_{23}\big(\gamma(R)\cdot (s_3s_4ab - s_1s_2cd)-\gamma(S)\cdot (r_3r_4ab - r_1r_2cd)\big), \\
    \beta_2 &:= h_1(R)\cdot h_2(S) - h_2(R)\cdot h_1(S).
\end{align*}
Then, maintaining the notation above and letting $D_{ij} := D_i\cdot D_j$, we find
\begin{align*}
    c_1 &= 2\beta_1  h_1(R)  h_1(S), \\
    c_2 &= \beta_1\big(h_1(R)  h_2(S) + h_2(R)  h_1(S)\big) 
    \\
    & \quad {} + \beta_2\big(\gamma(R)  (s_3s_4ab-s_1s_2cd)+\gamma(S)  (r_3r_4ab-r_1r_2cd)\big)D_{23}, \\
    c_3 &= \beta_1\big(h_1(R)  h_3(S) + h_3(R)  h_1(S)\big) 
    \\
    & \quad {} + \beta_2\big(\gamma(R)  (s_2s_4ac-s_1s_3bd)+\gamma(S)  (r_2r_4ac-r_1r_3bd)\big)D_{13}, \\
    c_4 &= \beta_1\big(h_1(R)  h_4(S) + h_4(R)  h_1(S)\big)
    \\
    & \quad {} + \beta_2\big(\gamma(R)  (s_2s_3ad-s_1s_4bc)+\gamma(S)  (r_2r_3ad-r_1r_4bc)\big)D_{12}, \\
    c_5 &= 2\beta_2  \gamma(S)   \gamma(R).
\end{align*}
Note that the $c_1, \dots, c_5$ are symmetric in $R$ and $S$, as one would expect.

\begin{remark}
    The $(3,3)$-isogeny is defined over the field of definition of the fundamental constants $a,b,c,d$ of $\K$ 
    and the kernel generators $R,S$ (rather than the subgroup $\langle R, S \rangle$).
\end{remark}

\subsection{Evaluating points under the \texorpdfstring{$(3,3)$}{(3,3)}-isogeny.}\label{subsec:iso33eval}
Consider the $(3,3)$-isogeny $\phi: \K \rightarrow \K'$ and assume the coefficients $c_1, \dots, c_5$ have been computed. 
Given a point $P = (x_1 \colon x_2 \colon x_3 \colon x_4) \in \K$, the image $\phi(P) = (x'_1 \colon x'_2 \colon x'_3 \colon x'_4)$ is given by
\begin{align*}
    x'_1 := x_1(c_1x_1^2 + c_2x_2^2 + c_3x_3^2 + c_4x_4^2) + c_5x_2x_3x_4,\\
    x'_2 := x_2(c_2x_1^2 + c_1x_2^2 + c_4x_3^2 + c_3x_4^2) + c_5x_1x_3x_4,\\
    x'_3 := x_3(c_3x_1^2 + c_4x_2^2 + c_1x_3^2 + c_2x_4^2) + c_5x_1x_2x_4,\\
    x'_4 := x_4(c_4x_1^2 + c_3x_2^2 + c_2x_3^2 + c_1x_4^2) + c_5x_1x_2x_3.
\end{align*}
The fundamental theta constants of the image surface $\K^\prime$ can be computed in the same way, i.e., as $\phi((a\colon b \colon c \colon d)).$
Via Equation~\eqref{eq:EFGH}, we can then compute the constants $E',F',G',H'$ defining the equation of the surface $\K'$.

\subsection{Implementation}
\label{subsec:33implementation}
We implemented \((3,3)\)-isogeny evaluation using the formul\ae{} above.
We give explicit operation counts for $\field = \FF_q$, which will be necessary for our cryptographic application in~\cref{sec:hashfunction}.
To optimise the computation, we implement the following algorithms:
\begin{enumerate}
    \item \textsf{TriplingConstantsFromThetas}: given fundamental theta constants $(a\colon b\colon c \colon d)$, compute \emph{tripling constants} consisting of: 
    \begin{itemize}
        \item their inverses $(1/a\colon \allowbreak 1/b\colon 1/c \colon 1/d)$;
        \item their squares $(a^2\colon b^2\colon c^2 \colon d^2)$;
        \item squared dual theta constants $(A^2\colon B^2\colon C^2 \colon \allowbreak D^2)$; and
        \item their inverses $(1/A^2\colon 1/B^2 \colon 1/C^2 \colon 1/D^2)$.
    \end{itemize}
    For $\field = \FF_q$, this requires $12 \Fqmult, 4 \Fqsquare$, and $6 \Fqadd$.
    \item \textsf{Compute33Coefficients}: given coordinates of $R,S$ and the tripling constants, compute the coefficients $c_1, \dots, c_5$ defining the $(3,3)$-isogeny. When $\field = \FF_q$, this requires $76 \Fqmult, 8 \Fqsquare$, and $97 \Fqadd$.
    \item \textsf{Isogeny33Evaluate}: given the coefficients $c_1,\dots,c_5$, computes the image of a point $P\in\K$ under the corresponding $(3,3)$-isogeny (as explained in~\cref{subsec:iso33eval}). For $\field = \FF_q$, this requires $26 \Fqmult, 4 \Fqsquare$ and $16 \Fqadd$.
    \item \textsf{ComputeImageThetas}: given coefficients $c_1, \dots, c_5$ and the tripling constants, compute the fundamental theta constants defining the image curve. For $\field = \FF_q$, this requires $26 \Fqmult$ and $16 \Fqadd$.
\end{enumerate}
Details of the implementation can be found in the accompanying code.

\subsection{A note on \texorpdfstring{$(5,5)$}--isogenies on fast Kummers}\label{subsec:5-5}

Suppose we now have a $(5,5)$-subgroup of $\J[5]$, which induces a $(5,5)$-isogeny $\phi$ on the corresponding fast Kummer surface $\K = \J/\{\pm 1\}$ with kernel $G = \langle R,S \rangle$ for some $R,S \in \K[5]$ (i.e., $G$ is the image of the $(5,5)$-subgroup in $\K$).

Following the method in~\cref{sec:generalmethod}, we first compute the $G$-invariant quintic forms.
Let $B^R_{i,j}$ and $B^S_{i,j}$ be as before (where now $R,S$ are the $5$-torsion points), and define
\begin{align*}
    B^{2R}_{ij}(X_1, X_2, X_3, X_4) &\colonequals B_{i,j}(X_1, X_2, X_3, X_4; 2R), \\
    B^{2S}_{ij}(X_1, X_2, X_3, X_4) &\colonequals B_{i,j}(X_1, X_2, X_3, X_4; 2S).
\end{align*}
By~\cref{lemma:invaraintforms} (and following Flynn~\cite{flynn2015descent}), the quintic forms invariant under translation by $R$ are given by
\begin{align*}
	F^R_{ijklm} \colonequals \ &X_iB^R_{jk}B^{2R}_{lm} + X_iB^R_{jl}B^{2R}_{km} + X_iB^R_{jm}B^{2R}_{kl} + \\ 
    &X_iB^R_{kl}B^{2R}_{jm} + X_iB^R_{km}B^{2R}_{jl} + X_iB^R_{lm}B^{2R}_{jk} + \\
    & X_jB^R_{ik}B^{2R}_{lm} + X_jB^R_{il}B^{2R}_{km} + X_jB^R_{im}B^{2R}_{kl} + \\ 
    &X_jB^R_{kl}B^{2R}_{im} + X_jB^R_{km}B^{2R}_{il} + X_jB^R_{lm}B^{2R}_{ik} + \\
    &X_kB^R_{ji}B^{2R}_{lm} + X_kB^R_{jl}B^{2R}_{im} + X_kB^R_{jm}B^{2R}_{il} + \\ 
    &X_kB^R_{il}B^{2R}_{jm} + X_kB^R_{im}B^{2R}_{jl} + X_kB^R_{lm}B^{2R}_{ji} + \\
    &X_lB^R_{jk}B^{2R}_{im} + X_lB^R_{ji}B^{2R}_{km} + X_lB^R_{jm}B^{2R}_{ki} + \\ 
    &X_lB^R_{ki}B^{2R}_{jm} + X_lB^R_{km}B^{2R}_{ji} + X_lB^R_{im}B^{2R}_{jk} +\\
    &X_mB^R_{jk}B^{2R}_{li} + X_mB^R_{jl}B^{2R}_{ki} + X_mB^R_{ji}B^{2R}_{kl} + \\&X_mB^R_{kl}B^{2R}_{ji} + X_mB^R_{ki}B^{2R}_{jl} + X_mB^R_{li}B^{2R}_{jk},
\end{align*}
where $1 \leq i,j,k,l,m \leq 4$. We similarly define $F^S_{ijklm}$: the quintic forms invariant under translation by $S$.

Let $\mathcal{F}_R = \{F^R_{ijklm}\}$ and $\mathcal{F}_R = \{F^S_{ijklm}\}$.
Working modulo the equation defining the Kummer surface $\K$, 
the quintic forms in $\mathcal{F}_R$ and $\mathcal{F}_S$ each generate a space of 
dimension 12, for which we choose a basis  
\begin{align*}
    \Big\{ \, &F^R_{14444}, \allowbreak \ F^R_{23333}, \allowbreak \ F^R_{23334}, \allowbreak \ F^R_{23344}, \allowbreak \ F^R_{23444}, \allowbreak \ F^R_{24444}, \allowbreak \\ 
    & F^R_{33333}, \allowbreak \ F^R_{33334},
    F^R_{33344}, \allowbreak \ F^R_{33444}, \allowbreak \ F^R_{34444}, \allowbreak \ F^R_{44444} \, \Big\},
\end{align*}
and similarly for $\mathcal{F}_S$. 
These spaces intersect in a space of dimension 4, which gives a description of the $(5,5)$-isogeny. 
Finding the final scaling map to put the image into the correct form is left as future work.

\section{Generating \texorpdfstring{$(N^k,N^k)$}{(Nk,Nk)}-subgroups}\label{sec:kernels}

In the remaining two sections, we turn to building a hash function based on the $(3,3)$-isogenies derived in the previous section. The hash function will compute $(3^k,3^k)$-isogenies as chains of $(3,3)$-isogenies, and will start each such chain by computing a $(3^k, 3^k)$-subgroup on a fast Kummer surface. This section describes how such $(N^k, N^k)$-subgroups can be computed (for prime number $N$ and integer $k \geq 1$) in a way that is amenable to efficient and secure cryptographic implementations. 

We do this in two steps: first, in~\cref{subsec:gens} we show how to compute a \emph{symplectic} basis for the $N^k$-torsion on the Jacobian $\J$, which we then push down to the corresponding fast Kummer surface $\K = \J/\{\pm 1\}$; and second, we use this basis to compute the generators $R, S$ of the $(N^k,N^k)$-subgroup $G$ using the \emph{three-dimensional differential addition chain} introduced in~\cref{subsec:3dchains}. 

\subsection{Generating a symplectic basis on \texorpdfstring{$\K$}{K}}
\label{subsec:gens}

We first compute a symplectic basis for the $N^k$-torsion on the Jacobian $\J$ of the genus-2 curve $C_{\lambda,\mu,\nu}$ corresponding to $\K$. 
We compute this basis by generating $N^k$-torsion points on $\J$ until the $N^k$-Weil pairing condition given in~\cref{def:symplectic} is satisfied.\footnote{In our implementation, we compute the $N^k$-Weil pairing of points on $\J[N^k]$ using {\tt MAGMA}'s in-built functionality.} These points are then pushed down to $\K$ via 
\begin{align*}
\pi \colon & \J \rightarrow \K, \\
	& P \mapsto \left(\mathcal{C}_{\mathcal{I}(\mathcal{O}_\mathcal{K})} \circ \mathcal{H} \circ\mathcal{C}_{(\mathcal{I}\circ \mathcal{H} \circ \mathcal{S})(\mathcal{O}_\mathcal{K})} \mathcal{S} \circ \mathcal{H} \circ \kappa \right)(P).
\end{align*}
Here, $\mathcal{C}$, $\mathcal{H}$, $\mathcal{S}$ and $\mathcal{I}$ are the standard Kummer operations defined in~\cref{sub:operations} and $\kappa: \J \rightarrow \K^{\rm Sqr}$ maps points in Mumford coordinates on $\J$ to the \emph{squared} Kummer surface $\K^{\rm Sqr}$ as follows. For generic points 
$P = (x^2+u_1x+u_0, v_1x+v_0)\in \J$, we have 
\[ \kappa :  (x^2+u_1x+u_0,v_1x+v_0) \mapsto \left( X_1 \colon X_2 \colon X_3 \colon X_4 \right)
\]
with 
\begin{align*}
    X_1&=a^2(u_0(\mu-u_0)(\lambda+u_1+\nu)-v_0^2)  ,
    &X_2&=b^2(u_0(\nu\lambda-u_0)(1+u_1+\mu)-v_0^2),\\  
    X_3&=c^2(u_0(\nu-u_0)(\lambda+u_1+\mu)-v_0^2)  ,
    &X_4&=d^2(u_0(\mu\lambda-u_0)(1+u_1+\nu)-v_0^2).
\end{align*}
For special points $P = (x - u_0, v_0)$, the map $\kappa$ is defined by first adding a point of order $2$ on $\J$. Indeed, by adding $(x- u'_0, 0) \in \J[2]$ (with $u'_0 \neq u_0$) we get the point $$\left (x^2 - (u_0 + u'_0) x + u_0u'_0, \frac{v_0}{u_0-u'_0}x - \frac{v_0u'_0}{u_0-u'_0} \right ),$$ to which we can apply $\kappa$. The translation is then undone by applying the action of the corresponding $2$-torsion point on $\K^{\rm Sqr}$.
Finally, $\kappa(\OO_\J) := (a^2 \colon b^2 \colon c^2 \colon d^2)$. 

The map $\kappa$ is due to Bisson, Cosset and Robert~\cite{AVIsogenies},
while the subsequent operations that map from $\K^{\rm Sqr}$ to $\K$
appear in Renes and Smith~\cite[Section 4.3]{qDSA}. Note that the map $\pi$ corresponds to a $(2,2)$-isogeny, which will not affect the order of the basis points if $N$ is coprime to $2$. If, however, $2 \mid N$ then one must additionally check the order of the image points on $\K$.

\begin{remark}\label{rem:precompute}
    For applications where \(N\), \(k\), and the domain Kummer $\K$ are
    fixed, the symplectic basis for $\J[N^k]$ can be computed as part of the set-up once and for all, and the image of these basis points (under $\pi$) can be hardcoded as system parameters. For instance, this will be the case for our cryptographic application in~\cref{sec:hashfunction}. In these scenarios, optimising the efficiency of the operations in this subsection is not a priority; the important goal is to optimise the efficiency of the \emph{online} part of the $(N^k,N^k)$-subgroup generation procedure, which amounts to optimising the three-dimensional differential addition chain in the following subsection. 
\end{remark}

\subsection{Three-dimensional differential addition chains}\label{subsec:3dchains}

Let $Q_1$, $Q_2$, $Q_3$, $Q_4 \in \K$ be the images of a symplectic basis for $\J[N^k]$ under the map $\pi$ described above. In this subsection we show how to use this basis to compute the two generators $R$ and $S$ of our $(N^k,N^k)$-subgroup. 

As a first simplification, we restrict to $(N^k,N^k)$-subgroups with generators of the form
\begin{equation}\label{eq:kernelgens}
	\begin{aligned}
		R = Q_1 + [\alpha] Q_3 + [\beta] Q_4, \\
    	S = Q_2 + [\beta] Q_3 + [\gamma] Q_4,
	\end{aligned}
\end{equation}
where $\alpha, \beta, \gamma \in \ZZ/{N^k}\ZZ$. There are $N^{3k}$ such $(N^k,N^k)$-subgroups (for example, see~\cite[Table 1]{secretkeys}), and there are $O(N^{3k-1})$ subgroups that we lose by imposing this restriction~\cite[Def. 3]{secretkeys}. In other words, at least half of the $(N^k,N^k)$-subgroups can be obtained with kernel generators of the form in~\eqref{eq:kernelgens}, so in a cryptographic context we lose at most one bit of security by simplifying in this manner. 

The remainder of this subsection presents the \textsf{3DAC} algorithm that allows us to compute the kernel generators $R$ and $S$ via Equation~\eqref{eq:kernelgens}. Our task is to define an algorithm that computes $P_1+[\beta]P_2+[\gamma]P_3$ for given scalars $\beta, \gamma \in \ZZ/{N^k}\ZZ$ and for the points $P_1$, $P_2$ and $P_3$ on $\K$. The analogous computation on $\J$ could utilise a straightforward 3-way multiexponentation algorithm, but on the Kummer surface we need a three-dimensional \emph{differential addition chain}. Such an addition chain is only allowed to include the computation of the sum $Q+R$ if the difference $\pm (Q-R)$ has already been computed at a previous stage. 

Three-dimensional differential addition chains have been studied previously by Rao~\cite{3diffchain} and more generally by Hutchinson and Karabina~\cite{dMUL}. However, both of those works study the general scenario whereby three scalars are in play. Viewing Equation~\eqref{eq:kernelgens}, we see that there is no scalar multiplication of the point $P_1$ in our case, which allows for some convenient simplifications. Moreover, the chain due to Rao~\cite{3diffchain} is non-uniform and the chain due to Hutchinson and Karabina~\cite{dMUL} is not fixed length unless the input scalars are. Our algorithm \textsf{3DAC} satisfies both of these properties regardless of the input scalars, making it secure for use in cryptographic applications, such as the hash function we present in~\cref{sec:hashfunction}. We remark that these properties would not be necessary for a hash function involving only public data, but for other cryptographic applications where inputs to the hash function (or, more broadly, the scalars $\beta$ and $\gamma$) are secret, such as key derivation functions, these properties are an imperative first-step towards protecting the secret data from side-channel attacks. 

We derived the \textsf{3DAC} chain by extending the two-dimensional differential addition chain due to Bernstein~\cite{bernstein2006differential}, the fastest known two-dimensional differential addition chain that is fixed length and uniform. Beyond the points $P_1$, $P_2$, $P_3$, \textsf{3DAC} also needs seven additional combinations of sums and/or differences of these three points. We specify the full ten-tuple of inputs as 
\begin{multline*}
    \mathcal{D} := \big(
        P_1, P_2, P_3, P_2 + P_3, P_2 - P_3, P_1 - P_2,
    P_1-P_3, [2](P_2 + P_3), 
    \\
        P_1 + P_2 + P_3, P_1 - P_2 - P_3
        \big),
\end{multline*}
for points \(P_1\), \(P_2\), \(P_3\) $\in \K$. 
In line with~\cref{rem:precompute}, these additional sums and differences can be (pre)computed on $\J$ and their image in $\K$ specified as part of the system parameters. 

Given the tuple $\mathcal{D}$ above and the two scalars $\beta,
\gamma \in \ZZ/{N^k}\ZZ$, our \textsf{3DAC} algorithm computes the point
$P_1+[\beta]P_2+[\gamma]P_3$ on $\K$ using $3\ell-2$ pseudo-additions
and $\ell-1$ pseudo-doublings on $\K$, 
where
$\ell$ is the bitlength of $N^k$.
See~\cref{appendix:3dac} for the full description of the algorithm.

\section{
    A hash function from \texorpdfstring{$(3,3)$}{(3,3)}-isogenies
}
\label{sec:hashfunction}

Isogenies between \emph{superspecial} Jacobians of genus 2 curves have been proposed for use in post-quantum isogeny-based cryptography (e.g.,~\cite{castryck2020hash,g2SIDH}). We follow suit and henceforth restrict our attention superspecial Jacobians defined over $\fieldbar = \Fpbar$.

\begin{definition}
    We say that the Jacobian $\J$ of a genus 2 curve is \emph{superspecial} if the Hasse--Witt matrix $M \in \Fp^{2\times 2}$ vanishes identically. We say that a Kummer surface $\K = \J/\{\pm 1 \}$ is \emph{superspecial} if the corresponding Jacobian $\J$ is superspecial. 
\end{definition}
It can be shown that every superspecial $\J/\Fpbar$ is $\Fpbar$-isomorphic to a Jacobian defined over $\Fpsq$. Similarly,
the corresponding superspecial Kummer surface $\K/\Fpbar$ is $\Fpbar$-isomorphic to a Kummer surface with model defined over $\Fpsq$. 

As an application to exhibit our algorithms, we construct a fundamental cryptographic primitive: a hash function. The first isogeny-based hash function was introduced by Charles, Goren and Lauter who use isogenies between supersingular ellipitic curves~\cite{CGLhash}. 
The use of higher dimensional isogenies between superspecial Jacobians of genus-2 curves to construct a variant of the Charles--Goren--Lauter (CGL) hash function was previously explored by Castryck, Decru and Smith~\cite{castryck2020hash} using $(2,2)$-isogenies. They argue that although the computation of higher dimensional isogenies is more expensive, breaking the security of the hash function requires $\widetilde{O}(p^{3/2})$ time, rather than $\widetilde{O}(p^{1/2})$ time as in the CGL hash function. Therefore, smaller parameters can be used to obtain the same security. 
Following this, Castryck and Decru~\cite{castryck2021multiradical} use multiradical formulae for $(3,3)$-isogenies to construct such a hash function and obtain an asymptotic speed-up of around a factor of $9$.  
Later work by Decru and Kunzweiler~\cite{DecruKunzweiler23} construct a hash function using $(3,3)$-isogenies between general Kummer surfaces by improving on the formulae given by Bruin, Flynn and Testa~\cite{bruin2014descent}.

In this section, we describe a variant of the CGL hash function~\cite{CGLhash}, called \ourhash, that uses the formulae introduced in~\cref{sec:3-3} to compute chains of $(3,3)$-isogenies between fast Kummer surfaces. 
We obtain a speed-up of around $8$ -- $9{\tt x}$, and around $32$ -- $34{\tt x}$ compared to the Castryck--Decru and Decru--Kunzweiler hash functions, respectively, for security levels $\lambda = 128, 192$ and $256$. 
\subsection{Chains of \texorpdfstring{$(3,3)$}{(3,3)}-isogenies}

We first present the \textsf{Isogeny33Chain} routine (\cref{alg:33isochain}) for computing chains of $(3,3)$-isogenies. Though we use it as a building block for a hash function, we note that it can also be used in a variety of cryptographic applications. In particular, we have been careful to ensure that each component of the algorithm is amenable to constant-time cryptographic software. 

\subsubsection*{Tripling algorithm} We start by presenting \textsf{TPL} (\cref{alg:triple}), the algorithm to compute $[3]P$ from a point $P\in \K$ and the associated tripling constants (as defined in~\cref{subsec:33implementation}). This algorithm requires $26\Fqmult$, $12\Fqsquare$ and $32\Fqadd$. 

\begin{algorithm}
	\caption{\textsf{TPL}($P, {\tt TC}$):}\label{alg:triple}
	\begin{flushleft} 
	\textbf{Input:} Point $P\in \K$ and tripling constants ${\tt TC}$ (with $i$-th entry denoted ${\tt TC}_i$) \\
	\textbf{Output:} Point $Q \in \K$ where $Q = [3]P$.
	\end{flushleft} 
	\begin{algorithmic}[1]
    \STATE $R \leftarrow \mathcal{S}(P)$
    \STATE $R \leftarrow \mathcal{H}(R)$
    \STATE $Q \leftarrow \mathcal{S}(R)$
    \STATE $Q \leftarrow \mathcal{C}_{{\tt TC}_5}(Q)$
    \STATE $Q \leftarrow \mathcal{H}(Q)$
    \STATE $Q \leftarrow \mathcal{C}_{{\tt TC}_4}(Q)$
	\STATE $Q \leftarrow \mathcal{S}(Q)$
	\STATE $Q \leftarrow \mathcal{H}(Q)$
	\STATE $S \leftarrow \mathcal{C}_{{\tt TC}_5}(R)$
    \STATE $Q \leftarrow \mathcal{C}_S(Q)$
    \STATE $Q \leftarrow \mathcal{H}(Q)$
    \STATE $Q \leftarrow \mathcal{C}_{\mathcal{I}(P)}(Q)$ 
	\RETURN $Q$
	\end{algorithmic}
\end{algorithm}

\subsubsection*{Naïve strategies} Given a $(3^k, 3^k)$-subgroup $G = \langle R, S \rangle \subset \K[3^k]$, we use \textsf{TPL} and the algorithms from~\cref{subsec:33implementation} to compute an isogeny with kernel $\langle R, S \rangle$ as a chain of $(3,3)$-isogenies of length $k$. 
A naïve way of doing so is the following. Set $P_0 := R$, $Q_0 := S$ and $\K_0 := \K$, and then execute the following four steps for $i = 1$ to $k$:
\begin{enumerate}
	\item Compute the tripling constants on $\K_{i-1}$ using \textsf{TriplingConstantsFromThetas}
	\item Compute $3$-torsion points $(P_i,Q_i):= (3^{k-i}R, 3^{k-i}S)$ using $k-i$ repeated applications of \textsf{TPL} on $R$ and $S$. 
	\item Compute the $(3,3)$-isogeny $\varphi_i: \K_{i-1}\rightarrow \K_i$ with kernel $\langle P_i,Q_i \rangle$, and the images of $P_{i-1},Q_{i-1}$ under this isogeny using \textsf{Compute33Coefficients} and \textsf{Isogeny33Evaluate}.
	\item Compute the theta constants of the image $\K_{i}$ of $\varphi_i$ using \textsf{ComputeImageThetas}.
\end{enumerate}
The $(3^k, 3^k)$-isogeny with kernel $G$ will be given by $\varphi_k \circ \dots \circ \varphi_1: \K_0 \rightarrow \K_{k}$.

\subsubsection*{Optimal strategies} A more efficient way to compute isogeny chains is to use \emph{optimal strategies}~\cite{sidh}. This allow us to reduce the number of executions of \textsf{TPL} needed to compute the kernel at each step in the chain by storing intermediate points obtained during the triplings and pushing them through each isogeny. In our case, the cost of tripling is around $1.1${\tt x} the cost of computing the image of a point under the isogeny, and so we shift the cost in this way to obtain the strategies (see~\cite{sidh} for further details on optimising this approach). 
We give this algorithm in detail in~\cref{alg:33isochain}, and note that invoking the optimal strategies results in a $3.7$--$6.6{\tt x}$ reduction of the cost to compute a chain of $(3,3)$-isogenies for the set of parameters we specify in~\cref{subsec:implementation}.
\begin{algorithm}
	\caption{\textsf{Isogeny33Chain}($k, \OO_{\K},R,S, {\tt strategy}$):}\label{alg:33isochain}
	\begin{flushleft} 
	\textbf{Input:} Fundamental theta constants $\OO_{\K} = (a \colon b\colon c\colon d)$ defining Kummer surface $\K$, generators $R,S \in \K$ of $(3^k,3^k)$-subgroup for $k \geq 1$, and optimal strategy ${\tt strategy}$. \\
	\textbf{Output:} Fundamental theta constants defining image Kummer surface $\K'$ of $(3^k,3^k)$-isogeny $\varphi: \K \rightarrow \K'$ with $\ker \varphi = \langle R,S \rangle$.
	\end{flushleft} 
	\begin{algorithmic}[1]
    \FOR{$e = k-1$ to $1$}
		\STATE $P, Q = R, S$
		\STATE ${\tt pts} = [ \ ]$
		\STATE ${\tt inds} = [ \ ]$
		\STATE $i \leftarrow 0$
        \STATE ${\tt TC} \leftarrow \textsf{TriplingConstantsFromThetas}(\OO_{\K})$
		\WHILE{$i < k-e$}
			\STATE Append $[R, S]$ to ${\tt pts}$
			\STATE Append $i$ to ${\tt inds}$
			\STATE $m \leftarrow {\tt strategy}[k-i-e+1]$
			\FOR{$j = 1$ to $m$}
				\STATE $R \leftarrow \textsf{TPL}(R, {\tt TC})$
				\STATE $S \leftarrow \textsf{TPL}(S, {\tt TC})$
			\ENDFOR
			\STATE $i \leftarrow i + m$
		\ENDWHILE
        \STATE ${\tt cs} \leftarrow \textsf{Compute33Coefficients}(P,Q, {\tt TC})$
        \STATE $\OO_{\K} \leftarrow \textsf{ComputingImageThetas}({\tt cs})$
        \FOR{$[P_1, P_2]$ in ${\tt pts}$}
			\STATE $P_1 \leftarrow \textsf{Isogeny33Evaluate}(P_1, {\tt cs})$
			\STATE $P_2 \leftarrow \textsf{Isogeny33Evaluate}(P_2, {\tt cs})$
		\ENDFOR
		\IF{${\tt pts}$ not empty}
			\STATE $[R, S] \leftarrow {\tt pts}[-1]$
			\STATE $i \leftarrow {\tt inds}[-1]$
			\STATE Remove last element from {\tt pts} and {\tt inds}
		\ENDIF
    \ENDFOR
	\STATE ${\tt TC} \leftarrow \textsf{TriplingConstantsFromThetas}(\OO_\K)$
	\STATE ${\tt cs} \leftarrow \textsf{Compute33Coefficients}(R,S, {\tt TC})$
	\STATE $\OO_{\K} \leftarrow \textsf{ComputingImageThetas}({\tt cs})$
	\RETURN $\OO_{\K}$
	\end{algorithmic}
\end{algorithm}

\subsection{A cryptographic hash function}

We are now ready to present the hash function \ourhash that uses chains of $(3,3)$-isogenies between fast Kummer surfaces. 
For a fixed security parameter $\secp$ and working over characteristic $p \approx 2^\lambda$, 
the hash function parses the message into three scalars $\alpha, \beta, \gamma \in \ZZ/{3^k}\ZZ$ which are fed into the \textsf{3DAC} algorithm to compute a $(3^k, 3^k)$-subgroup $G = \langle R,S \rangle$. This is then used to compute the corresponding $(3^k, 3^k)$-isogeny $\varphi: \K \rightarrow \K'$, and the output of the hash function is the fundamental theta constants of the image surface $\K'$.

To optimise our hash function, we want to ensure that the $(3^k, 3^k)$-isogeny is $\Fpsq$-rational. Given a security parameter $\lambda$, we choose a suitably sized prime $p = 16f\cdot 3^k -1 \approx 2^\lambda$
(the fact that $16\mid p+1$ ensures rational $2$-torsion), where $f$ is a small cofactor, and take a small pseudo-random walk\footnote{In our implementation, we used \texttt{MAGMA}'s inbuilt Richelot isogeny routine to take 20 $(2,2)$-isogenies away from $C_0$.} from the superspecial Jacobian of the curve $C_0: y^2 = x^6+1$ to arrive at our starting Jacobian $\J = \Jac(C)$. Since $\J(\Fpsq) \cong (\ZZ/(p+1)\ZZ)^4$, the Rosenhain invariants of $C$ are all rational in $\Fpsq$ and we use these to compute $\OO_{\K} = (a \colon b \colon c \colon d)$, i.e. the fundamental theta constants of the starting Kummer surface $\K = \J/\{\pm 1 \}$. Using the methods described in~\cref{subsec:gens}, we compute a symplectic basis for $\J[3^k]$, which (together with the auxiliary sums and differences defined in~\cref{subsec:3dchains}) is pushed down to our starting fast Kummer surface  $\K$ via $\pi: \J \rightarrow \K$ to obtain the two tuples of points $\mathcal{D}_R$ and $\mathcal{D}_S$. The setup routine outputs ${\tt gens} = \{\mathcal{D}_R, \ \mathcal{D}_{S}\}$ and ${\tt data} = \{k, \ \lceil \lambda/\log(3) \rceil, \ \OO_{\K}\}$.

The hash function takes as input {\tt data}, {\tt gens} and a message\footnote{In practice, the input message {\tt msg} to the hash function would be a bit string, which would then be parsed into the scalars $\alpha, \beta, \gamma$.}  ${\tt msg}= (\alpha, \beta, \gamma)$, where $\alpha, \beta, \gamma \in \ZZ/3^k\ZZ$, and the output of the hash function is a tuple of elements $\Fpsq$, namely the fundamental theta constants of the image Kummer surface $\K'$ under the $(3^k,3^k)$-isogeny defined by scalars $(\alpha, \beta, \gamma)$. The output is of size $8\log(p)$ without normalising the theta constants $(a',  b',  c', d')$, and of size $6\log(p)$ with normalisation $(a'/d', b'/d', c'/d')$, which comes at a cost of one inversion and $3\Fqmult$. 
We fully specify \ourhash in~\cref{alg:KuHash}.

We remark that, as the message must be parsed into three scalars lying in 
$\ZZ/3^k\ZZ$, our hash function requires an input message of 
length $3k\log_2(3)$ bits. To allow for arbitrary length messages,  
after each $(3^k, 3^k)$-isogeny 
is computed, we must sample a new symplectic basis on the image Kummer surface. 
We can do this efficiently 
using~\cite[Alg. 1]{secretkeys-full} (the full version of~\cite{secretkeys}).
However, we emphasize that our main objective in presenting the hash 
function \ourhash is to benchmark our algorithm \textsf{Isogeny33Chain} for 
computing chains of $(3,3)$-isogenies
against others in the literature. We therefore restrict 
our implementation and experiments to messages of length $3k\log_2(3)$.

\begin{algorithm}
	\caption{$\ourhash({\tt msg}, {\tt data}, {\tt gens})$}\label{alg:KuHash}
	\begin{flushleft} 
	\textbf{Input:} A message {\tt msg}, auxiliary data {\tt data} and generators {\tt gens} of $\K[3^k]$. \\
	\textbf{Output:} Fundamental theta constants $(a',b',c',d')$ of image Kummer surface $\K'$
	\end{flushleft} 
	\begin{algorithmic}[1]
    \STATE Parse {\tt msg} as $\alpha, \beta, \gamma$.
	\STATE Parse {\tt data} as $k, \ell, \mathcal{O}_{\K}$.
	\STATE Parse {\tt gens} as two sets $\mathcal{D}_{R}$, $\mathcal{D}_{S}$ (see~\cref{subsec:3dchains}).
	\STATE $R \leftarrow \textsf{3DAC}(\mathcal{D}_{R}, \alpha,\beta,\ell,\OO_{\K})$
	\STATE $S \leftarrow \textsf{3DAC}(\mathcal{D}_{S}, \beta,\gamma,\ell,\OO_{\K})$
	\STATE $(a' \colon b \colon c'\colon d') \leftarrow \textsf{Isogeny33Chain}(k, \OO_{\K}, R,S)$
	\RETURN $(a',b',c',d')$
	\end{algorithmic}
\end{algorithm}

\subsection{Implementation}\label{subsec:implementation}

We implement \ourhash and give parameters for security levels $\lambda = 128, 192$, and $256$. 

\subsubsection*{Security}

Let $\secp$ be the security parameter and $p \approx 2^\lambda$. We follow the discussion in~\cite[\S 7.4]{castryck2020hash} to determine the security of the hash function \ourhash. In particular, the security of our hash function is not affected by taking $N=3$ rather than $N=2$, and as a result the security of our hash function relies on similar problems, namely~\cref{prob:sec1} and~\cref{prob:sec2} below.

\begin{problem}\label{prob:sec1}
	Given two superspecial genus 2 curves $C_1, C_2$ defined over $\Fpsq$ find a $(3^k, 3^k)$-isogeny between $\Jac(C_1)$ and $\Jac(C_2)$.
\end{problem}

\begin{problem}\label{prob:sec2}
	Given a superspecial genus 2 curve $C_1$ defined over $\Fpsq$, find
	\begin{itemize}
		\item a curve $C_2$ and a $(3^k, 3^k)$-isogeny $\Jac(C_1) \rightarrow \Jac(C_2)$, 
		\item a curve $C_2'$ and a $(3^{k'}, 3^{k'})$-isogeny $\Jac(C_1) \rightarrow \Jac(C'_2)$, 
	\end{itemize}
	such that $\Jac(C_2)$ and $\Jac(C'_2)$ are $\Fpbar$-isomorphic. Here, we can have $k=k'$, but the kernels of the corresponding isogenies must be different. 
\end{problem}
Previous works take the general Pollard-$\rho$ attack to be the best classical attack against these problems, which runs in $\widetilde{O}(p^{3/2})$. We take a more conservative approach and consider the Costello--Smith algorithm~\cite{costellosmith} to be the best classical attack, which runs in $\widetilde{O}(p)$. 
The best quantum attack is based on Grover's claw-finding algorithm and runs in $\widetilde{O}(p^{1/2})$~\cite[Theorem 2]{costellosmith}.

Considering these attacks, we obtain the following parameters:
\begin{itemize}
	\item $\lambda = 128$ \ : \  $p = 5\cdot 2^4\cdot 3^k-1$ with $k = 75$;
	\item $\lambda = 192$ \ : \ $p = 37\cdot 2^4\cdot 3^k-1$ with $k = 115$;
	\item $\lambda = 256$ \ : \ $p = 11\cdot 2^4\cdot 3^k-1$ with $k = 154$.
\end{itemize}

\subsubsection*{Cost Metric} To benchmark \ourhash, we count $\Fp$-operations. Indeed, our implementation in Python/SageMath will call underlying $\Fp$-operations to compute $\Fpsq$-operations. 
For simplicity, our cost metric will take $\Fqmult = \Fqsquare$ and ignore $\Fqadd$, as additions have only a very minor impact on performance. 
Note that it is relatively straightforward to convert this cost into a more fine-grained metric (e.g., bit operations, cycle counts, etc.).

\subsubsection*{Results}

We ran \ourhash in SageMath version 10.1 using Python 3.11.1 and record the cost, as per the cost metric above, averaging over 100 random inputs for each prime size. We present the results in~\cref{tab:results}. 
Taking the $\Fpsq$-operation count from~\cite[\S 3.2]{DecruKunzweiler23} and using our cost metric, the cost of computing the coefficients of Decru and Kunzweiler's $(3,3)$-isogeny is $6702$ (assuming 1 $\Fpsq$-multiplication is equivalent to 3 $\Fp$-multiplications). We use this to obtain a lower bound on the cost of the Decru--Kunzweiler hash function. Though this is a lower bound, we see in~\cref{tab:results} that the cost already exceeds the total cost of \ourhash. 

The codebase for the other $(3,3)$-isogeny CGL hash variant by Castryck and Decru~\cite{castryck2021multiradical} uses Gröbner basis calculations, which is not realistic to convert to a fixed number of $\Fp$-operations. Therefore, for fair comparison, we ran all three hash functions in {\tt MAGMA} V2.25-6 on Intel(R) Core\texttrademark \ i7-1065G7 CPU @ 1.30GHx $\times$ 8 with 15.4 GiB memory, and record the time taken to run the hash functions for the different $\secp$ in~\cref{tab:results}. We again average over 100 random inputs for each prime size. 

\begin{table}[ht]
	\centering 
	\renewcommand{\tabcolsep}{0.08cm}
	\renewcommand{\arraystretch}{1.5}
	\begin{tabular}{c|c|c|c|c|c|}
		\cline{2-6}
	& $\lambda$ & Message Length & Cost          & Time (s) & Time per input bit (ms) \\ \hline
		\multicolumn{1}{|c|}{\multirow{3}{*}{\ourhash}} 
		  & $128$                  & $225\log_2(3)$ & $177956$      & 0.18     & 0.50                                                               \\
		\multicolumn{1}{|c|}{}   & $192$                  & $345\log_2(3)$ & $286636$      & 0.29     & 0.53             \\
		\multicolumn{1}{|c|}{}     & $256$                  & $462\log_2(3)$ & $396942$      & 0.53     & 0.72       \\ \hline
		\multicolumn{1}{|c|}{\multirow{3}{*}{\cite{DecruKunzweiler23}}}        & $128$    & $240\log_2(3)$ & $ > 536160$ & 5.99     & 15.75     \\
		\multicolumn{1}{|c|}{}       & $171$                  & $315\log_2(3)$ & $> 703710$  & 10.16     & 20.35     \\
		\multicolumn{1}{|c|}{}           & $256$                  & $477\log_2(3)$ & $> 1065618$ & 18.29    & 24.19   \\ \hline
		\multicolumn{1}{|c|}{\multirow{3}{*}{\cite{castryck2021multiradical}}} & $128$    & $357$          & -             & 1.51     & 4.23   \\
		\multicolumn{1}{|c|}{}  & $171$    & $547$          & -             & 2.82     & 5.16    \\
		\multicolumn{1}{|c|}{}    & $256$   & $732$     & -             & 4.81    & 6.57   \\ \hline
		\end{tabular}
	\vspace{0.3cm}
	\caption{Comparison of cost using cost metric and time taken to run \ourhash and hash functions in~\cite{castryck2021multiradical} and~\cite{DecruKunzweiler23}. All results are averaged over 100 runs with random inputs. We remark that the cost of \ourhash is the same for all runs because it is uniform.}\label{tab:results}
\end{table}
Comparing the time taken per input bit for $\lambda = 128$ ad $256$, we observe a speed-up of around $8$ -- $9{\tt x}$ compared to the Castryck--Decru hash function, and around $32$ -- $34{\tt x}$ compared to the Decru--Kunzweiler hash function. For a precise comparison between implementations, however, exact $\Fp$-operation counts of the Castryck--Decru and Decru--Kunzweiler hash functions are required. 
We note that an advantage of the algorithms developed with our approach is that we do not rely on in-built functionality and all our algorithms are uniform. Therefore, we are able to give precise $\Fp$-operation counts for \ourhash.

\sloppy \printbibliography \fussy

\newpage 

\appendix
\section{
    Explicit \texorpdfstring{$(2,2)$}{(2,2)}-isogenies on fast Kummers
}
\label{appendix:2-2}

The $15$ unique $(2,2)$-subgroups are given by 
\begingroup
    \allowdisplaybreaks
\begin{align*}
    G_{1,2} &:= \{(a \colon b \colon c \colon d), \ (a \colon b \colon -c \colon -d), \ (a \colon -b \colon c \colon -d), \ (a \colon -b \colon -c \colon d) \}, \\
    G_{1,4} &:= \{(a \colon b \colon c \colon d), \ (a \colon b \colon -c \colon -d), \ (b \colon a \colon d \colon c), \ (b \colon a \colon -d \colon -c)\}, \\
    G_{1,6} &:= \{(a \colon b \colon c \colon d), \ (a \colon b \colon -c \colon -d), \ (b \colon -a \colon d \colon -c), \ (b \colon -a \colon -d \colon c) \}, \\
    G_{2,8} &:= \{(a \colon b \colon c \colon d), \ (a \colon -b \colon c \colon -d), \ (c \colon d \colon a \colon b), \ (c \colon -d \colon a \colon -b)\}, \\
    G_{2,9} &:= \{(a \colon b \colon c \colon d), \ (a \colon -b \colon c \colon -d), \ (c \colon d \colon -a \colon -b), \ (c \colon -d \colon -a \colon b) \}, \\
    G_{3,12} &:= \{(a \colon b \colon c \colon d), \ (a \colon -b \colon -c \colon d), \ (d \colon c \colon b \colon a), \ (d \colon -c \colon -b \colon a) \}, \\
    G_{3,14} &:= \{(a \colon b \colon c \colon d), \ (a \colon -b \colon -c \colon d), \ (d \colon c \colon -b \colon -a), \ (d \colon -c \colon b \colon -a) \}, \\
    G_{4,8} &:= \{(a \colon b \colon c \colon d),(b \colon a \colon d \colon c),(c \colon d \colon a \colon b),(d \colon c \colon b \colon a) \}, \\
    G_{4,9} &:= \{(a \colon b \colon c \colon d),(b \colon a \colon d \colon c),(c \colon d \colon -a \colon -b),(d \colon c \colon -b \colon -a) \}, \\
    G_{5,10} &:= \{(a \colon b \colon c \colon d),(b \colon a \colon -d \colon -c),(c \colon -d \colon a \colon -b),(d \colon -c \colon -b \colon a) \}, \\
    G_{5,11} &:= \{(a \colon b \colon c \colon d),(b \colon a \colon -d \colon -c),(c \colon -d \colon -a \colon b),(d \colon -c \colon b \colon -a) \}, \\
    G_{6,8} &:= \{(a \colon b \colon c \colon d),(b \colon -a \colon d \colon -c),(c \colon d \colon a \colon b),(d \colon -c \colon b \colon -a) \}, \\
    G_{6,9} &:= \{(a \colon b \colon c \colon d),(b \colon -a \colon d \colon -c),(c \colon d \colon -a \colon -b),(d \colon -c \colon -b \colon a) \}, \\
    G_{7,10} &:= \{(a \colon b \colon c \colon d),(b \colon -a \colon -d \colon c),(c \colon -d \colon a \colon -b),(d \colon c \colon -b \colon -a) \}, \\
    G_{7,11} &:= \{(a \colon b \colon c \colon d),(b \colon -a \colon -d \colon c),(c \colon -d \colon -a \colon b),(d \colon c \colon b \colon a)\}.
\end{align*}
\endgroup
For each $(2,2)$-subgroup $G_{i,j}$,
the corresponding $(2,2)$-isogeny $\K \rightarrow \widetilde{\K}$ is given by 
\[
    \psi_{i,j} := \mathcal{H} \circ \mathcal{S} \circ \alpha
\]
where $\alpha: \K \rightarrow \K$ 
is induced by a linear map on $\mathbb{A}^4$.
If $i$ is a root of $x^2+1$ in $\fieldbar[x]$,
then the maps $\alpha$
for each $(2,2)$-subgroup $G_{i,j}$ are
defined by the following matrices:
\begingroup
    \allowdisplaybreaks
\begin{align*}
    G_{1,2} &: \left(\begin{smallmatrix}
        1 & 0 &0 & 0 \\
        0 & 1 &0 & 0 \\
        0 & 0 &1 & 0 \\
        0 & 0 &0 & 1
    \end{smallmatrix}\right)
    &
    G_{1,4} &: \left(\begin{smallmatrix}
        1 & 1 &0 & 0 \\
        1 & -1 &0 & 0 \\
        0 & 0 &1 & 1 \\
        0 & 0 &1 & -1
    \end{smallmatrix}\right)
    &
    G_{1,6} &: \left(\begin{smallmatrix}
        1 & i &0 & 0 \\
        1 & -i &0 & 0 \\
        0 & 0 &1 & i \\
        0 & 0 &1 & -i
    \end{smallmatrix}\right)
    \\
    G_{2,8} &: \left(\begin{smallmatrix}
        1 & 0 &1 & 0 \\
        1 & 0 &-1 & 0 \\
        0 & 1 &0 & 1 \\
        0 & 1 &0 & -1
    \end{smallmatrix}\right),
    &
    G_{2,9} &: \left(\begin{smallmatrix}
        1 & 0 &i & 0 \\
        1 & 0 &-i & 0 \\
        0 & 1 &0 & i \\
        0 & 1 &0 & -i
    \end{smallmatrix}\right),
    &
    G_{3,12} &: \left(\begin{smallmatrix}
        1 & 0 &0 & 1 \\
        1 & 0 &0 & -1 \\
        0 & 1 &1 & 0 \\
        0 & 1 &-1 & 0
    \end{smallmatrix}\right),
    \\
    G_{3,14} &: \left(\begin{smallmatrix}
        1 & 0 &0 & i \\
        1 & 0 &0 & -i \\
        0 & 1 &i & 0 \\
        0 & 1 &-i & 0
    \end{smallmatrix}\right),
    &
    G_{4,8} &: \left(\begin{smallmatrix}
        1 & 1 &1 & 1 \\
        1 & 1 &-1 & -1 \\
        1 & -1 &1 & -1 \\
        1 & -1 &-1 & 1
    \end{smallmatrix}\right),
    &
    G_{4,9} &: \left(\begin{smallmatrix}
        1 & 1 & i & i \\
        1 & 1 &-i & -i \\
        1 & -1 &i & -i \\
        1 & -1 &-i & i
    \end{smallmatrix}\right),
    \\
    G_{5,10} &: \left(\begin{smallmatrix}
        -1 & 1 &1 & 1 \\
        1 & -1 &1 & 1 \\
        1 & 1 &-1 & 1 \\
        1 & 1 &1 & -1
    \end{smallmatrix}\right),
    &
    G_{5,11} &: \left(\begin{smallmatrix}
        1 & -1 &-i & -i \\
        1 & -1 &i & i \\
        1 & 1 &-i & i \\
        1 & 1 &i & -i
    \end{smallmatrix}\right),
    &
    G_{6,8} &: \left(\begin{smallmatrix}
        1 & i &1 & i \\
        1 & i &-1 & -i \\
        1 & -i &1 & -i \\
        1 & -i &-1 & i
    \end{smallmatrix}\right),
    \\
    G_{6,9} &: \left(\begin{smallmatrix}
        1 & -i &-i & -1 \\
        1 & -i &i & 1 \\
        1 & i &-i & 1 \\
        1 & i &i & -1
    \end{smallmatrix}\right),
    &
    G_{7,10} &: \left(\begin{smallmatrix}
        1 & -i &-1 & -i \\
        1 & -i &1 & i \\
        1 & i &-1 & i \\
        1 & i &1 & -i
    \end{smallmatrix}\right),
    &
    G_{7,11} &: \left(\begin{smallmatrix}
        1 & i &i & 1 \\
        1 & i &-i & -1 \\
        1 & -i &i & -1 \\
        1 & -i &-i & 1
    \end{smallmatrix}\right).
\end{align*}
\endgroup

\section{Three-dimensional Addition Chain}\label{appendix:3dac}

We describe the three-dimensional addition chain \textsf{3DAC} needed to compute kernel generators of the form
$P_1 + [\beta] P_2 + [\gamma] P_3$.
We obtained this chain by extending Bernstein's two-dimensional differential addition chain~\cite{bernstein2006differential}, the fastest known chain of this kind that is both fixed length and uniform. 
\textsf{3DAC}
computes $P_1 + [\beta] P_2 + [\gamma] P_3$,
given
scalars $\beta, \gamma \in \ZZ/{3^k}\ZZ$,
the length $\ell$ of the chain, the
fundamental theta constants $\OO_{\K}$ of a fast Kummer surface $\K$,
and a tuple of points 
\begin{multline}
    \label{eq:3dac-input-tuple}
    \mathcal{D} := \big(P_1, P_2, P_3, P_2 + P_3, P_2 - P_3, P_1 - P_2,
    P_1-P_3, [2](P_2 + P_3),
    \\
    P_1 + P_2 + P_3, P_1 - P_2 - P_3 \big)
\end{multline}
on $\K$.
\textsf{3DAC} requires three subroutines: \textsf{DBLTHRICEADD},
\textsf{ENCODE}, and \textsf{IND}.

\(\textsf{DBLTHRICEADD}\) is defined to compute the mapping 
\begin{align*}
    \big(P,Q,P-Q, R, S, R-S, T,U, T-U\big) \longmapsto
    \big([2]P,P+Q,R+S,T+U\big)
\end{align*}
for points $P$, $Q$, $R$, $S$, $T$, and $U$ in $\K$,
using one pseudo-doubling and three pseudo-additions on $\K$. 

\textsf{ENCODE} takes two $\ell$-bit scalars $\beta, \gamma \in
\ZZ/{3^k}\ZZ$ and outputs a single bit \textsf{b} and four
$(\ell-1)$-bit scalars $b_i$ for $i = 0,1,2,3$. The bit \textsf{b}
selects the input for the differential addition that kickstarts the chain (see
Step~\ref{step:first} of~\cref{alg:3DAC}), after which the $j$-th bits
(for $j =1, \dots, \ell$) of each of the four $b_i$ determine one of 16 input permutations that is fed into the $\textsf{DBLTHRICEADD}$ routine (see Step~\ref{step:switch} of~\cref{alg:3DAC}). 

\begin{algorithm}
	\caption{\textsf{ENCODE}($\beta$, $\gamma$):}\label{alg:ENCODE}
	\begin{flushleft} 
	\textbf{Input:}  two $\ell$-bit scalars $\beta=(\beta[\ell-1], \dots , \beta[0])$ and $\gamma=(\gamma[\ell-1], \dots , \gamma[0])$\\
	\textbf{Output:} a bit ${\sf b} \in \{0,1\}$, and four $(\ell-1)$-bit scalars $(b_0,b_1,b_2,b_3)$
	\end{flushleft} 
	\begin{algorithmic}[1]
	\STATE ${\sf b} \leftarrow \beta[1]$
    \FOR{$i =1$ to $\ell-1$}
	\STATE $b_1[i] \leftarrow \beta[i] \oplus \beta[i+1]$ 
         \STATE $b_0[i] \leftarrow b_1[i] \oplus \gamma[i] \oplus \gamma[i+1]$
	\STATE $b_2[i] \leftarrow \beta[i+1] \oplus \gamma[i+1]$
	\STATE $b_3[i]\leftarrow {\sf b}$
	\STATE ${\sf b} \leftarrow b_1[i] \oplus (b_0[i] \oplus 1) \otimes {\sf b}$ 
    \ENDFOR
	\RETURN ${\sf b}$, $(b_0,b_1,b_2,b_3)$
	\end{algorithmic}
\end{algorithm}

The indexing algorithm \textsf{IND} is used to choose one of four points as the last input to the $\textsf{DBLTHRICEADD}$ algorithm. 

\begin{algorithm}
	\caption{\textsf{IND}($I$):}\label{alg:ind}
	\begin{flushleft} 
	\textbf{Input:}  a 6-tuple of integers $I=(I_1, \dots, I_6)$\\
	\textbf{Output:} an integer index $\textsf{ind} \in \{1,2,3,4\}$
	\end{flushleft} 
	\begin{algorithmic}[1]
	\SWITCH {$[I_3-I_1, I_4-I_2]$}
\CASELINE{$[-1,-1]$} $\textsf{ind} \leftarrow 1$	
\CASELINE{$[\enspace \, 1, \enspace \, \, 1]$} $\textsf{ind} \leftarrow 2$	
\CASELINE{$[\enspace \,  1,-1]$} $\textsf{ind} \leftarrow 3$	
\CASELINE{$[-1,\enspace \, 1]$} $\textsf{ind} \leftarrow 4$	
\ENDSWITCH
	\RETURN $\textsf{ind}$
	\end{algorithmic}
\end{algorithm}

\begin{algorithm}
	\caption{\textsf{3DAC}($\mathcal{D}, \beta, \gamma, \ell, \OO_{\K}$):}\label{alg:3DAC}
	\begin{flushleft} 
	\textbf{Input: } Scalars $\beta, \gamma \in  \ZZ/{3^k}\ZZ$, the
        length of chain $\ell$, fundamental theta constants $\OO_{\K}$,
        and a tuple $\mathcal{D}$ of points on $\K$ as in
        Eq.~\eqref{eq:3dac-input-tuple}. \\
	\textbf{Output: } $P_1 + [\beta] P_2 + [\gamma] P_3$ where $P_1, P_2, P_3$ are $\mathcal{D}_1,\mathcal{D}_2,\mathcal{D}_3$, respectively.
	\end{flushleft} 
	\begin{algorithmic}[1]
	\STATE initialise $P \leftarrow (\mathcal{D}_4, \mathcal{D}_8, \mathcal{D}_4, \mathcal{D}_9)$ \\ 
	\STATE initialise $D \leftarrow (\mathcal{D}_2, \mathcal{D}_3, \mathcal{D}_4, \mathcal{D}_5)$ \\ 
	\STATE initialise $\Delta \leftarrow (\mathcal{D}_1, \mathcal{D}_{10}, \mathcal{D}_6, \mathcal{D}_7)$ \\ 
	\STATE initialise $I \leftarrow (1,1,2,2,1,1)$ 
	\STATE ${\sf b}, (b_0,b_1,b_2,b_3) \leftarrow \textsf{ENCODE}(\beta, \gamma)$ \label{step:first}
	\IF{${\sf b} = 1$}
    	\STATE $(P_3, I_6) \leftarrow \left((P_3, D_2, D_1), I_6+1 \right)$
  	\ELSE 
	\STATE $(P_3, I_5) \leftarrow \left((P_3, D_1, D_2), I_5+1 \right)$ 
  	\ENDIF
 	\SWITCH {$(b_0, b_1, b_2, b_3)$} \label{step:switch}
\CASELINE{$(0,0,0,0)$} $I \leftarrow (I_1+I_3,I_2+I_4,2I_3,2I_4,I_3+I_5,I_4+I_6)$	\\
 $(P_2,P_1,P_3,P_4) \leftarrow {\sf DBLTHRICEADD}(P_2,P_1,D_3,P_3,P_2,D_2,P_2,P_4,\Delta_{\textsf{IND}(I)})$
\CASELINE{$(0,0,0,1)$} $I \leftarrow (I_1+I_3,I_2+I_4,2I_3,2I_4,I_3+I_5,I_4+I_6)$	\\
 $(P_2,P_1,P_3,P_4) \leftarrow {\sf DBLTHRICEADD}(P_2,P_1,D_3,P_3,P_2,D_1,P_2,P_4,\Delta_{\textsf{IND}(I)})$
\CASELINE{$(0,0,1,0)$} $I \leftarrow (I_1+I_3,I_2+I_4,2I_3,2I_4,I_3+I_5,I_4+I_6)$	\\
$(P_2,P_1,P_3,P_4) \leftarrow {\sf DBLTHRICEADD}(P_2,P_1,D_4,P_3,P_2,D_2,P_2,P_4,\Delta_{\textsf{IND}(I)})$
\CASELINE{$(0,0,1,1)$} $I \leftarrow (I_1+I_3,I_2+I_4,2I_3,2I_4,I_3+I_5,I_4+I_6)$	\\
 $(P_2,P_1,P_3,P_4) \leftarrow {\sf DBLTHRICEADD}(P_2,P_1,D_4,P_3,P_2,D_1,P_2,P_4,\Delta_{\textsf{IND}(I)})$
\CASELINE{$(0,1,0,0)$} $I \leftarrow (I_1+I_3,I_2+I_4,2I_1,2I_2,I_1+I_5,I_2+I_6)$	\\
 $(P_2,P_1,P_3,P_4) \leftarrow {\sf DBLTHRICEADD}(P_1,P_2,D_3,P_3,P_1,D_2,P_1,P_4,\Delta_{\textsf{IND}(I)})$
\CASELINE{$(0,1,0,1)$} $I \leftarrow (I_1+I_3,I_2+I_4,2I_1,2I_2,I_1+I_5,I_2+I_6)$	\\
 $(P_2,P_1,P_3,P_4) \leftarrow {\sf DBLTHRICEADD}(P_1,P_2,D_3,P_3,P_1,D_1,P_1,P_4,\Delta_{\textsf{IND}(I)})$
\CASELINE{$(0,1,1,0)$} $I \leftarrow (I_1+I_3,I_2+I_4,2I_1,2I_2,I_1+I_5,I_2+I_6)$	\\
  $(P_2,P_1,P_3,P_4) \leftarrow {\sf DBLTHRICEADD}(P_1,P_2,D_4,P_3,P_1,D_2,P_1,P_4,\Delta_{\textsf{IND}(I)})$
\CASELINE{$(0,1,1,1)$} $I \leftarrow (I_1+I_3,I_2+I_4,2I_1,2I_2,I_1+I_5,I_2+I_6)$	\\
  $(P_2,P_1,P_3,P_4) \leftarrow {\sf DBLTHRICEADD}(P_1,P_2,D_4,P_3,P_1,D_1,P_1,P_4,\Delta_{\textsf{IND}(I)})$
\CASELINE{$(1,0,0,0)$} $I \leftarrow (I_1+I_3,I_2+I_4,2I_5,2I_6,I_3+I_5,I_4+I_6)$	\\
 $(P_2,P_3,P_1,P_4) \leftarrow {\sf DBLTHRICEADD}(P_3,P_2,D_2,P_1,P_2,D_3,P_3,P_4,\Delta_{\textsf{IND}(I)})$
\CASELINE{$(1,0,0,1)$} $I \leftarrow (I_1+I_3,I_2+I_4,2I_5,2I_6,I_1+I_5,I_2+I_6)$	\\
 $(P_2,P_3,P_1,P_4) \leftarrow {\sf DBLTHRICEADD}(P_3,P_1,D_1,P_1,P_2,D_3,P_3,P_4,\Delta_{\textsf{IND}(I)})$
\CASELINE{$(1,0,1,0)$} $I \leftarrow (I_1+I_3,I_2+I_4,2I_5,2I_6,I_3+I_5,I_4+I_6)$	\\
$(P_2,P_3,P_1,P_4) \leftarrow {\sf DBLTHRICEADD}(P_3,P_2,D_2,P_1,P_2,D_4,P_3,P_4,\Delta_{\textsf{IND}(I)})$
\CASELINE{$(1,0,1,1)$} $I \leftarrow (I_1+I_3,I_2+I_4,2I_5,2I_6,I_1+I_5,I_2+I_6)$	\\
 $(P_2,P_3,P_1,P_4) \leftarrow {\sf DBLTHRICEADD}(P_3,P_1,D_1,P_1,P_2,D_4,P_3,P_4,\Delta_{\textsf{IND}(I)})$
\CASELINE{$(1,1,0,0)$} $I \leftarrow (I_1+I_3,I_2+I_4,2I_5,2I_6,I_1+I_5,I_2+I_6)$	\\
 $(P_2,P_3,P_1,P_4) \leftarrow {\sf DBLTHRICEADD}(P_3,P_1,D_2,P_1,P_2,D_3,P_3,P_4,\Delta_{\textsf{IND}(I)})$
\CASELINE{$(1,1,0,1)$} $I \leftarrow (I_1+I_3,I_2+I_4,2I_5,2I_6,I_3+I_5,I_4+I_6)$	\\
 $(P_2,P_3,P_1,P_4) \leftarrow {\sf DBLTHRICEADD}(P_3,P_2,D_1,P_1,P_2,D_3,P_3,P_4,\Delta_{\textsf{IND}(I)})$
\CASELINE{$(1,1,1,0)$} $I \leftarrow (I_1+I_3,I_2+I_4,2I_5,2I_6,I_1+I_5,I_2+I_6)$	\\
$(P_2,P_3,P_1,P_4) \leftarrow {\sf DBLTHRICEADD}(P_3,P_1,D_2,P_1,P_2,D_4,P_3,P_4,\Delta_{\textsf{IND}(I)})$
\CASELINE{$(1,1,1,1)$} $I \leftarrow (I_1+I_3,I_2+I_4,2I_5,2I_6,I_3+I_5,I_4+I_6)$	\\
  $(P_2,P_3,P_1,P_4) \leftarrow {\sf DBLTHRICEADD}(P_3,P_2,D_1,P_1,P_2,D_4,P_3,P_4,\Delta_{\textsf{IND}(I)})$
\ENDSWITCH
	\RETURN $P_4$
	\end{algorithmic}
\end{algorithm}

\end{document}